\documentclass[twoside,11pt]{article}
\usepackage{graphicx}                   
\usepackage{amsmath}                    
\usepackage{amssymb}                    
\usepackage{amsfonts}                   
\usepackage{mathrsfs}                   
\usepackage{url}                        
\usepackage{color}                      
\usepackage{hyperref}                   
\usepackage{enumitem}                   
\hypersetup{colorlinks=true,linkcolor=[rgb]{0.75,0,0},citecolor=[rgb]{0,0,0.75}}
\usepackage{amsthm}                     
\usepackage{thmtools}                   
\usepackage{fullpage}
\usepackage[parfill]{parskip}
\usepackage{comment}
\usepackage[backend=biber,
  maxnames=5,                           
  giveninits=true,                      
  sortcites                             
  ]{biblatex}
\renewbibmacro{in:}{}                   
\declaretheorem[name=Theorem]{theorem}
\declaretheorem[name=Lemma,numberwithin=section]{lemma}
\declaretheorem[name=Corollary,sibling=lemma]{corollary}

\newcommand{\ceil}[1]{\left\lceil #1\right\rceil}
\newcommand{\ang}[1]{\langle #1\rangle}
\newcommand{\half}[1]{\frac{#1}{2}}

\renewcommand{\ang}[1]{\langle #1\rangle}

\newcommand{\RE}{\mathbb{R}}            
\newcommand{\eps}{\varepsilon}          
\newcommand{\ST}{\,:\,}                 

\newcommand{\MM}{\mathscr{M}}

\newcommand{\QQ}{\mathcal{Q}}

\newcommand{\bd}{\partial}
\newcommand{\etal}{\textit{et al.}}

\DeclareMathOperator{\vol}{vol}
\DeclareMathOperator{\area}{area}
\DeclareMathOperator{\conv}{conv}

\DeclareMathOperator{\width}{wid}
\DeclareMathOperator{\ray}{ray}
\DeclareMathOperator{\dist}{dist}
\DeclareMathOperator{\base}{base}
\DeclareMathOperator{\interior}{int}

\DeclareMathOperator{\poly}{poly}
\newcommand{\dcap}[2]{\textup{dcap}_{#1}(#2)}
\newcommand{\pcap}[2]{\textup{cap}_{#1}(#2)}
\newcommand{\icone}[2]{\textup{icone}(#1,#2)}
\newcommand{\ocone}[2]{\textup{ocone}(#1,#2)}

\newtheorem*{theorem*}{Theorem}

\begin{document}


\title{Economical Convex Coverings and Applications\thanks{An earlier version of this paper appeared in the \textit{Proceedings of the 2023 Annual ACM-SIAM Symposium on Discrete Algorithms} (SODA), pp. 1834--1861, 2023.}}

 \author{%
 	Sunil Arya\thanks{Research supported by the Research Grants Council of Hong Kong, China under project numbers 16213219 and 16214721. The work of David Mount was supported by NSF grant CCF--1618866. The work of Guilherme da Fonseca was supported by the French ANR PRC grant ADDS (ANR-19-CE48-0005).}\\
 		Department of Computer Science and Engineering \\
 		The Hong Kong University of Science and Technology, Hong Kong\\
 		arya@cse.ust.hk 
 		\and
 	Guilherme D. da Fonseca\footnotemark[1]\\
 	    Aix-Marseille Universit\'{e} and LIS Lab, France\\
 		guilherme.fonseca@lis-lab.fr
 		\and
 	David M. Mount\footnotemark[1]\\
 		Department of Computer Science and 
 		Institute for Advanced Computer Studies \\
 		University of Maryland, College Park, Maryland \\
 		mount@umd.edu 
 }

\date{}

\maketitle

\vspace{-5pt} 

\begin{abstract}
Coverings of convex bodies have emerged as a central component in the design of efficient solutions to approximation problems involving convex bodies. Intuitively, given a convex body $K$ and $\eps > 0$, a \emph{covering} is a collection of convex bodies whose union covers $K$ such that a constant factor expansion of each body lies within an $\eps$ expansion of $K$. Coverings have been employed in many applications, such as approximations for diameter, width, and $\eps$-kernels of point sets, approximate nearest neighbor searching, polytope approximations with low combinatorial complexity, and approximations to the Closest Vector Problem (CVP).

It is known how to construct coverings of size $n^{O(n)} / \eps^{(n-1)/2}$ for general convex bodies in $\RE^n$. In special cases, such as when the convex body is the $\ell_p$ unit ball, this bound has been improved to $2^{O(n)} / \eps^{(n-1)/2}$. This raises the question of whether such a bound generally holds. In this paper we answer the question in the affirmative. 

We demonstrate the power and versatility of our coverings by applying them to the problem of approximating a convex body by a polytope, where the error is measured through the Banach-Mazur metric. Given a well-centered convex body $K$ and an approximation parameter $\eps > 0$, we show that there exists a polytope $P$ consisting of $2^{O(n)} / \eps^{(n-1)/2}$ vertices (facets) such that $K \subset P \subset K(1+\eps)$. This bound is optimal in the worst case up to factors of $2^{O(n)}$. (This bound has been established recently using different techniques, but our approach is arguably simpler and more elegant.) As an additional consequence, we obtain the fastest $(1+\eps)$-approximate CVP algorithm that works in any norm, with a running time of $2^{O(n)} / \eps^{(n-1)/2}$ up to polynomial factors in the input size, and we obtain the fastest $(1+\eps)$-approximation algorithm for integer programming. We also present a framework for constructing coverings of optimal size for any convex body (up to factors of $2^{O(n)}$). 
\end{abstract}


\section{Introduction} \label{s:intro}

Convex bodies are of fundamental importance in mathematics and computer science, and given the high complexity of exact representations, concise approximate representations are essential to many applications. There are a number of ways to define the distance between two convex bodies (see, e.g., \cite{Bor00}), and each gives rise to a different notion of approximation. While Hausdorff distance is commonly studied, it is not sensitive to the shape of the convex body. In this paper we will consider a common linear-invariant distance, called the Banach-Mazur distance. 

Given two convex bodies $X$ and $Y$ in real $n$-dimensional space, $\RE^n$, both of which contain the origin in their interiors, their \emph{Banach-Mazur distance}, denoted $\dist_{\text{BM}}(X,Y)$, is defined to be the minimum value of $\ln \lambda$ such that there exists a linear transformation $T$ such that $T X \subseteq Y \subseteq \lambda \cdot T X$. Given $\delta > 0$, we say that $Y$ is an \emph{Banach-Mazur $\delta$-approximation} of $X$ if $\dist_{\text{BM}}(X,Y) \leq \delta$. $T$ will be the identity transformation in our constructions, and thus, given a convex body $K$ in $\RE^n$ and $\eps > 0$, we seek a convex polytope $P$ such that $K \subseteq P \subseteq (1+\eps) K$. This implies that $\dist_{\text BM}(K,P) \leq \ln (1+\eps)$, which is approximately $\eps$ for small $\eps$. The scaling is taking place about the origin, and it is standard practice to assume that $K$ is well-centered in the sense that the origin lies within $K$ and is not too close to $K$'s boundary. (See Section~\ref{s:centrality} for the formal definition.) Unlike Hausdorff, the Banach-Mazur measure has the desirable property of being sensitive to $K$'s shape, being more accurate where $K$ is narrower and less accurate where $K$ is wider.

The principal question is, given $n$ and $\eps > 0$, what is the minimum number of vertices (or facets) needed to $\eps$-approximate any convex body $K$ in $\RE^n$ by a polytope in the above sense. This problem has been well studied. Existing bounds hold under the assumption that $K$ is well-centered. We say that a bound is \emph{nonuniform} if it holds for all $\eps \leq \eps_0$, where $\eps_0$ depends on $K$. Typical nonuniform bounds assume that $K$ is smooth, and the value of $\eps_0$ depends on $K$'s smoothness. Our focus will be on uniform bounds, where $\eps_0$ does not depend on $K$.

Dudley~\cite{Dud74} and Bronshtein and Ivanov~\cite{BrI76} provided uniform bounds in the Hausdorff context, but their results can be recast under Banach-Mazur, where they imply the existence of an approximating polytope with $n^{O(n)} / \eps^{(n-1)/2}$ vertices (facets). For smooth convex bodies, B{\"{o}}r{\"{o}}czky \cite{Bor00,Gru93b} established a nonuniform bound of $2^{O(n)} / \eps^{(n-1)/2}$. Barvinok~\cite{Bar14} improved the bound in the uniform setting for symmetric convex bodies. Ignoring a factor that is polylogarithmic in $1/\eps$, his bound is $2^{O(n)} / \eps^{n/2}$. Finally, Nasz{\' o}di, Nazarov, and Ryabogin obtained a worst-case optimal approximation of size $2^{O(n)} / \eps^{(n-1)/2}$~\cite{NNR20}. Their bound is uniform and holds for general convex bodies.

The main result of this paper is an alternative asymptotically optimal construction of an $\eps$-approximation of a convex body $K$ in $\RE^n$ in the Banach-Mazur setting. Our construction is superior to that of \cite{NNR20} in two ways. First, while the construction presented in \cite{NNR20} is very clever, it involves the combination of a number of technical elements (transforming the body to standard position, rounding it, computing a Bronshteın-Ivanov net, and filtering to reduce the sample size). In contrast, ours is quite simple. We employ a greedy process that samples points from $K$'s interior, and the final approximation is just the convex hull of these points. Second, our construction is more powerful in that it provides an additional covering structure for $K$. Each sample point is associated with a centrally symmetric convex body, and together these bodies form a cover of $K$ such that their union lies within the expansion $(1+\eps) K$. As a direct consequence of this additional structure, we obtain the fastest approximation algorithm to date for the closest vector problem (CVP) that operates in any norm.

\subsection{Our Results} \label{s:results}

Throughout, we assume that $K$ is a full-dimensional convex body in $\RE^n$, which is well-centered about the origin. There are a number of notions of centrality that suffice for our purposes (see Section~\ref{s:centrality} for formal definitions). Our first result involves the existence of concise coverings. Given a convex body $K$ that contains the origin in its interior and reals $c \geq 1$ and $\eps > 0$, a \emph{$(c,\eps)$-covering} of $K$ is a collection $\QQ$ of bodies whose union covers $K$ such that a factor-$c$ expansion of each $Q \in \QQ$ about its centroid lies within $(1+\eps) K$ (see Figure~\ref{f:cover-basic}). Coverings have emerged as an important tool in convex approximation. They have been applied to several problems in the field of computational geometry, including combinatorial complexity~\cite{AAFM22, AFM17c, AFM12b}, approximate nearest neighbor searching~\cite{AFM17a}, and computing the diameter and $\eps$-kernels~\cite{AFM17b}. 

\begin{figure}[htbp]
\centering
\includegraphics[scale=.8,page=2]{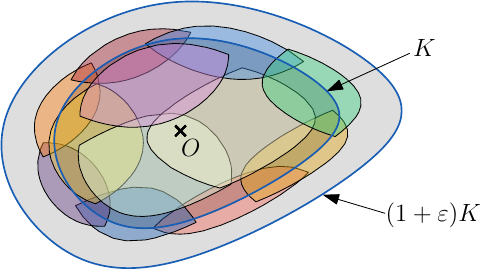}
\caption{\label{f:cover-basic} A $(2,\eps)$-covering.}
\end{figure}

Given a convex body in $\RE^n$, constant $c \geq 1$ and parameter $\eps > 0$, what is the minimum size of a $(c,\eps)$-covering as a function of $n$ and $\eps$? Abdelkader and Mount considered the problem in spaces of constant dimension~\cite{AbM18}. They did not analyze their bounds for the high-dimensional case, but based on results from \cite{AFM17a}, it can be shown that their results yield an upper bound of $n^{O(n)} / \eps^{(n-1)/2}$ in $\RE^n$. A number of special cases have been explored in the high dimensional case. Nasz{\' o}di and Venzin demonstrated the existence of $(2,\eps)$-coverings of size $2^{O(n)} / \eps^{n/2}$ when $K$ is an $\ell_p$ ball for any fixed $p \geq 2$~\cite{NaV22}. For the $\ell_{\infty}$ ball, Eisenbrand, H{\"a}hnle, and Niemeier showed the existence of $(2,\eps)$-coverings of size $2^{O(n)} / \log^n (1/\eps)$, consisting of axis-parallel rectangles~\cite{EHN11}. They also presented a nearly matching lower bound of $2^{-O(n)} / \log^n(1/\eps)$, even when the covering consisted of parallelepipeds. 

In this paper we establish the following bound on the size of $(c,\eps)$-coverings, which holds for any well-centered convex body in $\RE^n$.

\begin{theorem} \label{thm:cover-worst}
Let $0 < \eps \leq 1$ be a real parameter and $c \geq 2$ be a constant. Let $K \subseteq \RE^n$ be a well-centered convex body. Then there is a $(c,\eps)$-covering for $K$ consisting of at most $2^{O(n)} / \eps^{(n-1)/2}$ centrally symmetric convex bodies.
\end{theorem}

It is not difficult to prove a lower bound of $2^{-O(n)} / \eps^{(n-1)/2}$ on the size of any $(2,\eps)$-covering for Euclidean balls (see, e.g., Nasz{\'o}di and Venzin \cite{NaV22}). Therefore, the above bound is optimal with respect to $\eps$-dependencies. In Section~\ref{s:instance-opt} (Theorem~\ref{thm:cover-inst}), we prove that for any constant $c \geq 2$, our construction is in fact instance optimal to within a factor of $2^{O(n)}$. This means that for any well-centered convex body $K$, our covering exceeds the size of any $(c,\eps)$-covering for $K$ by such a factor. In Section~\ref{s:apps-cvp}, we present a randomized algorithm that constructs a slightly larger covering (by a factor of $\log(1/\eps)$). Following standard convention, our constructions assume that access to $K$ is provided by a weak membership oracle (defined in Section~\ref{s:apps}). 

We present a number of applications of this result. First, in Section~\ref{s:approx-BM} we show that the convex hull of the center points of the covering elements yields an approximation in the Banach-Mazur metric.

\begin{theorem} \label{thm:approx-BM}
Given a well-centered convex body $K$ and an approximation parameter $\eps > 0$, there exists a polytope $P$ consisting of $2^{O(n)} / \eps^{(n-1)/2}$ vertices (facets) such that $K \subset P \subset K(1+\eps)$.
\end{theorem}

There are also applications to lattice problems. In the \emph{Closest Vector Problem} (CVP), an $n$-dimensional lattice $L$ in $\RE^n$ is given (that is, the set of integer linear combinations of $n$ basis vectors) together with a target vector $t \in \RE^n$. The problem is to return a vector in $L$ closest to $t$ under some given norm. This problem has applications to cryptography~\cite{Odl90, JoS98, NgS01}, integer programming~\cite{Len83, DPV11, DaK16}, and factoring polynomials over the rationals~\cite{LLL82}, among several other problems. The problem is NP-hard for any $\ell_p$ norm~\cite{vEB81} and cannot be solved exactly in $2^{(1-\gamma)n}$ time for constant $\gamma > 0$, under certain conditional hardness assumptions~\cite{BGS17}. 

This problem has a considerable history. The first solution proposed to the CVP under the $\ell_\infty$ norm takes $2^{O(n^3)}$ time through integer linear programming~\cite{Len83}, which was later improved to $n^{O(n)}$~\cite{Kan87}. For the $\ell_2$ norm, Micciancio and Voulgaris presented an algorithm that runs in single exponential $2^{O(n)}$ time~\cite{MiV13}, and currently the fastest algorithm for exact Euclidean CVP is by Aggarwal, Dadush, and Stephens-Davidowitz~\cite{ADS15} and runs in $2^{n + o(n)}$ time. However, solving the CVP problem exactly in single exponential time for norms other than Euclidean remains an open problem. (For additional information, see~\cite{HPS11}.) Dadush, Peikert, and Vempala~\cite{DPV11} considered CVP and the related Shortest Vector Problem (SVP) in the context of (possibly asymmetric) norms defined by convex bodies. Their work demonstrated a rich connection between lattice algorithms and convex geometry.

In the approximate version of the CVP problem, denoted \emph{$(1+\eps)$-CVP}, we are also given a parameter $\eps > 0$, and the goal is to find a lattice vector whose distance to $t$ is at most $1+\eps$ times the optimum. CVP is NP-hard to approximate~\cite{Aro94, DKRS03} and conditional hardness results show that for $p \geq 1$ CVP in $\ell_p$ is hard to approximate in $2^{(1-\gamma)n}$ time for constant $\gamma > 0$, except when $p$ is even~\cite{ABGS21}. 

The randomized sieving approach of Ajtai, Kumar, and Sivakumar~\cite{AKS01} was extended to approximate CVP for $\ell_p$ norms by Bl{\"o}mer and Naewe~\cite{BN09} and to the general case of well-centered norms by Dadush~\cite{Dad14}. These algorithms run in time and space $2^{O(n)} / \eps^{2n}$. Building on the Voronoi cell approach~\cite{MiV13, DPV11}, Dadush and Kun~\cite{DaK16} presented deterministic algorithms that improved the running time to $2^{O(n)} / \eps^{n}$ and space to $\widetilde{O}(2^n)$. 

Eisenbrand, H{\"a}hnle, and Niemeier~\cite{EHN11} and Nasz{\'o}di and Venzin~\cite{NaV22} have explored the use of $(c,\eps)$-coverings of the unit ball in the norm to obtain efficient algorithms for approximate CVP by ``boosting'' a weak constant-factor approximation to a strong $(1+\eps)$-approximation. By exploiting the unique properties of hypercubes, Eisenbrand \etal~\cite{EHN11} improved the running time for the $\ell_\infty$ norm to $2^{O(n)} \log^n(1/\eps)$ time. Nasz{\'o}di and Venzin~\cite{NaV22} extended this approach to $\ell_p$ norms. The running time of their algorithm is $2^{O(n)} / \eps^{n/2}$ for $p \ge 2$ and $2^{O(n)} / \eps^{n/p}$ for $1 \le p \le 2$. The constants in the $2^{O(n)}$ term in the running time depend on $p$.

By applying our covering within existing algorithms, we obtain the fastest algorithm to date for $(1+\eps)$-approximate CVP that operates in any norm. The algorithm is randomized and runs in single exponential time, $2^{O(n)} / \eps^{(n-1)/2}$. (Following standard practice, we ignore factors that are polynomial in the input size.) The result is stated formally below.

\begin{theorem} \label{thm:cvp}
There is a randomized algorithm that, given any well-centered convex body $K$ and lattice $L$, solves the $(1+\eps)$-CVP problem in the norm defined by $K$, in $2^{O(n)} / \eps^{(n-1)/2}$-time and $O(2^{n})$-space, with probability at least $1 - 2^{-n}$.
\end{theorem}

Finally, through a reduction from approximate CVP to approximate integer programming (IP) due to Dadush~\cite{Dad14}, we present a randomized algorithm for approximate IP (see Theorem~\ref{thm:approx-ip} in Section~\ref{s:apps-ip}).

\subsection{Techniques} \label{s:techniques}

As mentioned above, coverings are a powerful tool in obtaining efficient solutions to approximation problems involving convex bodies. The fundamental problem tackled here involves the sizes of $(c,\eps)$-coverings for general convex bodies in $\RE^n$ and especially the dependencies on $\eps$. Our approach employs a classical concept from convex geometry, called a \emph{Macbeath region}~\cite{Mac52}. Given a convex body $K$ and a point $x \in K$, the Macbeath region $M_K(x)$ is the largest centrally symmetric body centered at $x$ and contained in $K$ (see Figure~\ref{f:macbeath}(a)). Macbeath regions have found numerous uses in the theory of convex sets and the geometry of numbers (see B\'{a}r\'{a}ny~\cite{Bar00} for an excellent survey). They have also been applied to several problems in the field of computational geometry, including lower bounds~\cite{BCP93, AMM09b, AMX12}, combinatorial complexity~\cite{AFM12b, MuR14, AFM17c, DGJ19, AAFM22}, approximate nearest neighbor searching~\cite{AFM17a}, and computing the diameter and $\eps$-kernels~\cite{AFM17b}.

\begin{figure}[htbp]
\centering
\includegraphics[scale=.8]{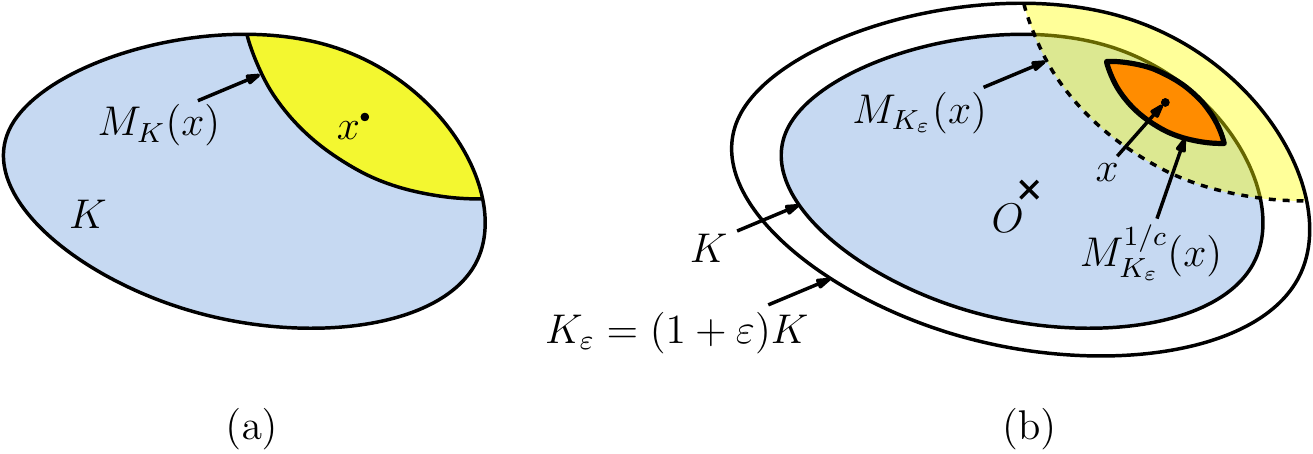}
\caption{\label{f:macbeath} (a) A Macbeath region and (b) a covering element derived from a shrunken Macbeath region.}
\end{figure}

In the context of $(c,\eps)$-coverings, the obvious (and indeed maximal) choice for a covering element centered at any point $x$ is to take the Macbeath region centered at $x$ with respect to the expanded body $K_{\eps} = (1+\eps) K$, and then scale it by a factor of $\frac{1}{c}$ about $x$ (see Figure~\ref{f:macbeath}(b)). The construction and analysis of such Macbeath-based coverings is among the principal contributions of this paper. In their work on the economical cap cover, B\'{a}r\'{a}ny and Larman observed how Macbeath regions serve as an efficient agent for covering the region near the boundary of a convex body~\cite{BaL88}. While Macbeath regions can be quite elongated, especially near the body's boundary, they behave in many respects like fixed-radius balls in a metric space. (Vernicos and Walsh proved that shrunken Macbeath regions are similar in shape to fixed-radius balls in the Hilbert geometry induced by $K$~\cite{AbM18, VeW16}.) This leads to a very simple covering construction based on computing a maximal set of points such that the suitably shrunken Macbeath regions centered at these points are pairwise disjoint. The covering is then constructed by uniformly increasing the scale factor so the resulting Macbeath regions cover $K$.

Two challenges arise in implementing and analyzing this construction. The first is that of how to compute these Macbeath regions efficiently. The second is proving that this simple construction yields the desired bound on the size of the covering. A natural approach to the latter is a packing argument based on volume considerations. Unfortunately, this fails because Macbeath regions may have very small volume. Our approach for dealing with small Macbeath regions is to exploit a Mahler-like reciprocal property in the volumes of the Macbeath regions in the original body $K$ and its polar, $K^*$ (see Section~\ref{s:centrality} for definitions). In the low-dimensional setting, the analysis exploits a correspondence between caps in $K$ and $K^*$, such that the volumes of these caps have a reciprocal relationship (see, e.g., \cite{AAFM22}). As a consequence, for each Macbeath region in $K$ of small volume, there is a Macbeath region in $K^*$ of large volume. Thus, by randomly sampling in both $K$ and $K^*$, it is possible to hit all the Macbeath regions. 

Generalizing this to the high-dimensional setting involves overcoming a number of technical difficulties. A straightforward generalization of the methods of \cite{AAFM22} yields a covering of size $n^{O(n)} / \eps^{(n-1)/2}$. A critical step in the analysis involves relating the volumes of two $(n-1)$-dimensional convex bodies that arise by projecting caps and dual caps. In earlier works, where the dimension was assumed to be a constant, a crude bound sufficed. But in the high-dimensional setting, it is essential to avoid factors that depend on the dimension. A key insight of this paper is that it is possible to avoid these factors through the use of the difference body. (See Lemma~\ref{lem:sandwich-dualcaps} in Section~\ref{s:diff-body}.) Through the use of this more refined geometric analysis, we establish this Mahler-like relationship in Sections~\ref{s:mahler} (particularly Lemmas~\ref{lem:vol-product} and~\ref{lem:mahler-mac}). We apply this in Section~\ref{s:worst-opt} to obtain our bounds on the size of the covering. In Section~\ref{s:approx-BM} we show how this leads to an $\eps$-approximation in the Banach-Mazur measure. The sampling process is described in Section~\ref{s:apps} along with applications.

\section{Preliminaries} \label{s:prelim}

In this section, we introduce terminology and notation, which will be used throughout the paper. This section can be skipped on first reading (moving directly to Section~\ref{s:mahler}).

\subsection{Lengths and Measures} \label{s:length}

Given vectors $u, v \in \RE^n$, let $\ang{u,v}$ denote their dot product, and let $\|v\| = \sqrt{\ang{v,v}}$ denote $v$'s Euclidean length. Throughout, we will use the terms \emph{point} and \emph{vector} interchangeably. Given points $p,q \in \RE^n$, let $\|p q\| = \|p - q\|$ denote the Euclidean distance between them. Let $\vol(\cdot)$ and $\area(\cdot)$ denote the $n$-dimensional and $(n-1)$-dimensional Lebesgue measures, respectively.

Throughout, $K \subseteq \RE^n$ will denote a full-dimensional compact convex body with the origin $O$ in its interior. Let $\|x\|_K = \inf \{s \ge 0: x \in s K\}$ denote $K$'s associated Minkowski functional, or \emph{gauge function}. If $K$ is centrally symmetric, its gauge function defines a norm, but we will abuse notation and use the term ``norm'' even when $K$ is not centrally symmetric. Given $\eps > 0$, define $K_{\eps} = (1+\eps)K$ to be a uniform scaling of $K$ by $1+\eps$.

Given a convex body $K \subseteq \RE^n$, its \emph{difference body}, denoted  $\Delta(K)$, is defined to be the Minkowski sum $K \oplus -K$. The difference body is convex and centrally symmetric and satisfies the following property.

\begin{lemma}[Rogers and Shephard~\cite{RoS57}] \label{lem:vol-diffbody}
Given a convex body $K \subseteq \RE^n$, $\vol(\Delta(K)) \le 4^n \vol(K)$.
\end{lemma}

\subsection{Polarity and Centrality Properties} \label{s:centrality}

Given a bounded convex body $K \subseteq \RE^n$ that contains the origin $O$ in its interior, define its \emph{polar}, denoted $K^*$, to be the convex set
\[
	K^*
		~ = ~ \{ u \ST \ang{u,v} \le 1, \hbox{~for all $v \in K$} \}.
\]
The polar enjoys many useful properties (see, e.g., Eggleston~\cite{Egg58}). For example, it is well known that $K^*$ is bounded and $(K^*)^* = K$. Further, if $K_1$ and $K_2$ are two convex bodies both containing the origin such that $K_1 \subseteq K_2$, then $K_2^* \subseteq K_1^*$. 

Given a nonzero vector $v \in \RE^n$, we define its ``polar'' $v^*$ to be the hyperplane that is orthogonal to $v$ and at distance $1/\|v\|$ from the origin, on the same side of the origin as $v$. The polar of a hyperplane is defined as the inverse of this mapping. We may equivalently define $K^*$ as the intersection of the closed halfspaces that contain the origin, bounded by the hyperplanes $v^*$, for all $v \in K$. 

Given a convex body $K \subseteq \RE^n$, there are many ways to characterize the property that $K$ is centered about the origin~\cite{Gru63, Tot15}. In this section we explore a few relevant measures of centrality.

First, define $K$'s \emph{Mahler volume} to be the product $\vol(K) \cdot \vol(K^*)$. The Mahler volume is well studied (see, e.g.~\cite{San49,MeP90,Sch93}). It is invariant under linear transformations, and it depends on the location of the origin within $K$. In the following definitions, any fixed constant may be used in the $O(n)$ term.

\begin{description}
    \item[Santal{\'o} property:] The Mahler volume of $K$ is at most $2^{O(n)} \cdot \omega_n^2$, where $\omega_n$ denotes the volume of the $n$-dimensional unit Euclidean ball ($\omega_n = \pi^{n/2} / \Gamma\big(\half{n}+1\big)$).
    
    \item[Winternitz property:] For any hyperplane passing through the origin, the ratio of the volume of the portion of $K$ on each side of the hyperplane to the volume of $K$ is at least $2^{-O(n)}$. 
    
    \item[Kovner-Besicovitch property:] The ratio of the volume of $K \cap -K$ to the volume of $K$ is at least $2^{-O(n)}$. 
\end{description}

Following Dadush, Peikert, and Vempala~\cite{DPV11}, we say that $K$ is \emph{well-centered} if it satisfies the Kovner-Besicovitch property. Generally, $K$ is \emph{well-centered} about a point $x$ if $K-x$ is well-centered. For our purposes, however, any of the above can be used, as shown in the following lemma. 

\begin{lemma}
\label{lem:centroid}
The three centrality properties (Santal{\'o}, Winternitz, and Kovner-Besicovitch) are equivalent in the sense that a convex body $K \subseteq \RE^n$ that satisfies any one of them satisfies the other two subject to a change in the $2^{O(n)}$ factor. Further, if the origin coincides with $K$'s centroid, these properties are all satisfied.
\end{lemma}

Let us first introduce some notation. Given a hyperplane $h$, let $h^+$ and $h^-$ denote its two halfspaces. Given $0 < \delta < \frac{1}{2}$, let $h$ be a hyperplane that intersects $K$ such that $\vol(K \cap h^+) = \delta \cdot \vol(K)$. Define the \emph{$\delta$-floating body}, denoted $K_\delta$, to be the intersection of halfspaces $h^-$ for all such hyperplanes $h$. For $t > 0$, define the \emph{$t$-Santal{\'o} region} $S(K,t) \subseteq K$ to be the set of points $x \in K$ such that the Mahler volume of $K$ with respect to $x$ is at most $t \, \omega_n^2$, where $\omega_{n}$ denotes the volume of the $n$-dimensional unit Euclidean ball. Both the floating body and the Santal{\'o} region (when nonempty) are convex subsets of $K$, and Meyer and Werner showed that they satisfy the following property.

\begin{lemma}[Meyer and Werner~\cite{MeW98}]
\label{lem:float-San49}
For all $0 < \delta < \frac{1}{2}$, $K_{\delta} \subseteq S(K,t)$, where $t = 1/(4\delta(1-\delta))$.
\end{lemma}

We also need the following result by Milman and Pajor~\cite{MiP00} (Remark~4 following Corollary~3), which implies that if $K$ satisfies Santal{\'o}, then it satisfies Kovner-Besicovitch. 

\begin{lemma}[Milman and Pajor~\cite{MiP00}] \label{lem:santalo-kb}
Let $K$ be a convex body with the origin $O$ in its interior such that $\vol(K) \cdot \vol(K^*) \leq s \, w_n^2$, where $s$ is a parameter. Then $\vol(K \cap -K) / \vol(K) \geq 2^{-O(n)} / s$.
\end{lemma}

We are now ready to prove Lemma~\ref{lem:centroid}.

\begin{proof} (of Lemma~\ref{lem:centroid})
First, suppose that $K$ satisfies Kovner-Besicovitch, that is, $\vol(K \cap -K) \ge 2^{-O(n)} \cdot \vol(K)$. Consider any hyperplane $h$ passing through the origin. As $K \cap -K$ is centrally symmetric, half of this body lies on each side of $h$. Thus, the volume of the portion of $K$ on either side of $h$ is at least $2^{-O(n)} \cdot \vol(K)$, and so $K$ satisfies the Winternitz property.

Next, suppose that $K$ satisfies Winternitz. Observe that any point outside the floating body $K_{\delta}$ is contained in a halfspace $h^+$ such that $\vol(K \cap h^+) \leq \delta \cdot \vol(K)$. By Winternitz, all halfspaces containing the origin have volume at least $2^{-O(n)} \cdot \vol(K)$, and so the origin is contained within the floating body $K_{\delta}$ for $\delta = 2^{-O(n)}$. It follows from Lemma~\ref{lem:float-San49} that the origin lies within the Santal{\'o} region $S(K,t)$ for some $t = 2^{O(n)}$. Thus, $K$ satisfies the Santal{\'o} property.

Finally, if $K$ satisfies Santal{\'o}, then it follows from Lemma~\ref{lem:santalo-kb} that it satisfies the Kovner-Besicovitch property. This establishes the equivalence of the three centrality properties.

Milman and Pajor~\cite{MiP00} (Corollary 3) showed that if the origin coincides with $K$'s centroid, then $K$ satisfies Kovner-Besicovitch, implying that it satisfies the other properties as well. 
\end{proof}

Lower bounds on the Mahler volume have also been extensively studied~\cite{BoM87,Kup08,Naz12}. Recalling the value of $\omega_n$ from the Santal{\'o} property, the following lower bound holds irrespective of the location of the origin within a convex body~\cite{BoM87}.

\begin{lemma}
\label{lem:mahler-bounds}
Given a convex body $K \subseteq \RE^n$ whose interior contains the origin, $\vol(K) \cdot \vol(K^*) \geq 2^{-O(n)} \cdot \omega_n^2$.
\end{lemma}

\subsection{Caps, Rays, and Relative Measures}

Consider a compact convex body $K$ in $n$-dimensional space $\RE^n$ with the origin $O$ in its interior. A \emph{cap} $C$ of $K$ is defined to be the nonempty intersection of $K$ with a halfspace. Letting $h_1$ denote a hyperplane that does not pass through the origin, let $\pcap{K}{h_1}$ denote the cap resulting by intersecting $K$ with the halfspace bounded by $h_1$ that does not contain the origin (see Figure~\ref{f:widray}(a)). Define the \emph{base} of $C$, denoted $\base(C)$, to be $h_1 \cap K$. Letting $h_0$ denote a supporting hyperplane for $K$ and $C$ parallel to $h_1$, define an \emph{apex} of $C$ to be any point of $h_0 \cap K$.

\begin{figure}[htbp]
\centering
\includegraphics[scale=.8]{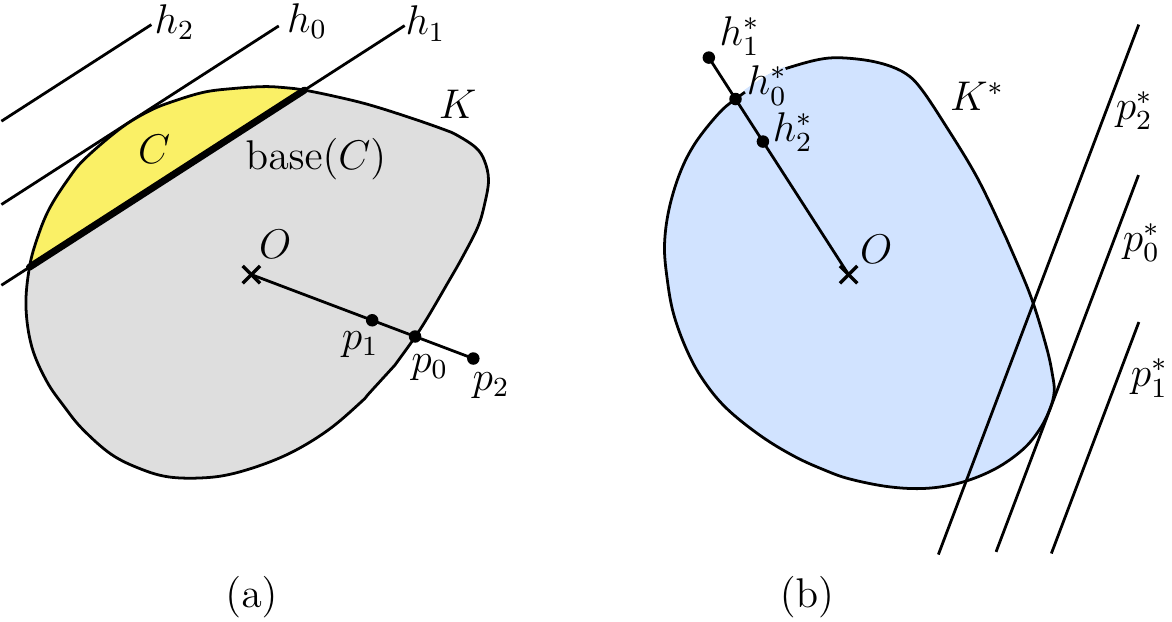}
\caption{\label{f:widray} Convex body $K$ and polar $K^*$ with definitions used for width and ray.}
\end{figure}

We define the \emph{absolute width} of cap $C$ to be $\dist(h_1,h_0)$. When a cap does not contain the origin, it will be convenient to define the \emph{relative width} of $C$, denoted $\width_K(C)$, to be the ratio $\dist(h_1,h_0) / \dist(O,h_0)$. We extend the notion of width to hyperplanes by defining $\width_K(h_1) = \width_K(\pcap{K}{h_1})$. Observe that as a hyperplane is translated from a supporting hyperplane to the origin, the relative width of its cap ranges from 0 to a limiting value of 1.

We also characterize the closeness of a point to the boundary in both absolute and relative terms. Given a point $p_1 \in K$, let $p_0$ denote the point of intersection of the ray $O p_1$ with the boundary of $K$. Define the \emph{absolute ray distance} of $p_1$ to be $\|p_1 p_0\|$, and define the \emph{relative ray distance} of $p_1$, denoted $\ray_K(p_1)$, to be the ratio $\|p_1 p_0\| / \|O p_0\|$. Relative widths and relative ray distances are both affine invariants, and unless otherwise specified, references to widths and ray distances will be understood to be in the relative sense.

We can also define volumes in a manner that is affine invariant. Recall that $\vol(\cdot)$ denotes the standard Lebesgue volume measure. For any region $\Lambda \subseteq K$, define the \emph{relative volume} of $\Lambda$ with respect to $K$, denoted $\vol_K(\Lambda)$, to be $\vol(\Lambda)/\vol(K)$.

With the aid of the polar transformation we can extend the concepts of width and ray distance to objects lying outside of $K$. Consider a hyperplane $h_2$ parallel to $h_1$ that lies beyond the supporting hyperplane $h_0$ (see Figure~\ref{f:widray}(a)). It follows that $h_2^* \in K^*$, and we define $\width_K(h_2) = \ray_{K^*}(h_2^*)$ (see Figure~\ref{f:widray}(b)). Similarly, for a point $p_2 \notin K$ that lies along the ray $O p_1$, it follows that the hyperplane $p_2^*$ intersects $K^*$, and we define $\ray_K(p_2) = \width_{K^*}(p_2^*)$. By properties of the polar transformation, it is easy to see that $\width_K(h_2) = \dist(h_0,h_2) / \dist(O,h_2)$. Similarly, $\ray_K(p_2) = \|p_0 p_2\| / \|O p_2\|$. Henceforth, we will omit references to $K$ when it is clear from context.

Some of our results apply only when we are sufficiently close to the boundary of $K$. Given $0 \leq \alpha \leq 1$, we say that a cap $C$ is \emph{$\alpha$-shallow} if $\width(C) \le \alpha$, and we say that a point $p$ is \emph{$\alpha$-shallow} if $\ray(p) \le \alpha$. We will simply say \emph{shallow} to mean $\alpha$-shallow, where $\alpha$ is a sufficiently small constant.

Given any cap $C$ and a real $\lambda > 0$, we define its $\lambda$-expansion, denoted $C^{\lambda}$, to be the cap of $K$ cut by a hyperplane parallel to the base of $C$ such that the absolute width of $C^{\lambda}$ is $\lambda$ times the absolute width of $C$. (Note that if the expansion of a cap is large enough it may be the same as $K$.) 

We now present a number of useful technical results on ray distances and cap widths in both their absolute and relative forms.

\begin{lemma} \label{lem:raydist-width}
Let $C$ be a cap of $K$ that does not contain the origin and let $p$ be a point in $C$. Then $\ray(p) \leq \width(C)$.
\end{lemma}

\begin{proof}
Let $h$ be the hyperplane passing through the base of $C$, and let $h_0$ be the supporting hyperplane of $K$ parallel to $h$ at $C$'s apex. Let $q$, $p_0$, and $q_0$ denote the points of intersection of the ray $O p$ with $h$, $\bd K$, and $h_0$, respectively. Since $p \in C$, the order of these points along the ray is $\ang{O, q, p, p_0, q_0}$. By considering the hyperplanes parallel to $h$ passing through these points, we have
\[
    \ray(p) 
        ~ =    ~ \frac{\|p p_0\|}{\|O p_0\|}
        ~ \leq ~ \frac{\|q p_0\|}{\|O p_0\|} 
        ~ \leq ~ \frac{\|q p_0\| + \|p_0 q_0\|}{\|O p_0\| + \|p_0 q_0\|} 
        ~ =    ~ \frac{\|q q_0\|}{\|O q_0\|} 
        ~ =    ~ \frac{\dist(h,h_0)}{\dist(O,h_0)}
        ~ =    ~ \width(C). \qedhere
\]
\end{proof}

There are two natural ways to associate a cap with any point $p \in K$. The first is the \emph{minimum volume cap}, which is any cap whose base passes through $p$ of minimum volume among all such caps. For the second, assume that $p \neq O$, and let $p_0$ denote the point of intersection of the ray $O p$ with the boundary of $K$. Let $h_0$ be any supporting hyperplane of $K$ at $p_0$. Take the cap $C$ induced by a hyperplane parallel to $h_0$ passing through $p$. As shown in the following lemma this is the cap of minimum width containing $p$.


\begin{lemma}
\label{lem:min-width-cap}
For any $p \in K \setminus \{O\}$, consider the cap $C$ defined above. Then $\width(C) = \ray(p)$ and further, $C$ has the minimum width over all caps that contain $p$.
 \end{lemma}


\begin{proof}
Let $h$ denote the hyperplane passing through $p$ parallel to $h_0$ (defined above). By similar triangles, we have 
\[
    \width(C) 
        ~ = ~ \frac{\dist(h,h_0)}{\dist(O,h_0)} 
        ~ = ~ \frac{\|p p_0\|}{\|O p_0\|} 
        ~ = ~ \ray(p).
\] 
By Lemma~\ref{lem:raydist-width}, for any cap $C'$ that contains $p$, $\ray(p) \leq \width(C')$, and hence $\width(C) \leq \width(C')$.
\end{proof}

The following lemma gives a simple lower and upper bound on the absolute volume of a cap.

\begin{lemma}
\label{lem:vol-cap}
Let $C$ be a $\frac{1}{2}$-shallow cap, let $a = \area(\base(C))$, and let $w$ denote $C$'s absolute width. Then $a w/n \leq \vol(C) \leq 2^{n-1} a w$.
\end{lemma}

\begin{proof}
Let $p$ be the apex of $C$ and $\base(C)$ denote its base. Let $P = \conv(\base(C) \cup \{p\})$. Clearly, $P \subseteq C$ and $\vol(P) = a w/n$, which yields the lower bound. To see the upper bound, observe that $C$ lies within the generalized infinite cone whose apex is $O$ and base is $\base(C)$. Because $\width(C) \leq \frac{1}{2}$, it follows that the area of any slice of $C$ cut by a hyperplane parallel to $\base(C)$ exceeds the area of $\base(C)$ by a factor of at most $2^{n-1}$. The upper bound follows from elementary geometry.
\end{proof}

An easy consequence of convexity is that, for $\lambda \ge 1$, $C^{\lambda}$ is a subset of the region obtained by scaling $C$ by a factor of $\lambda$ about its apex. This implies the following lemma.

\begin{lemma} \label{lem:cap-exp}
Given any cap $C$ and a real $\lambda \geq 1$, $\vol(C^{\lambda}) \leq \lambda^n \vol(C)$.
\end{lemma}

Another consequence of convexity is that containment of caps is preserved under expansion. This is a straightforward adaptation of Lemma~4.4 in~\cite{AFM17c}.

\begin{lemma} \label{lem:cap-containment-exp}
Given two caps $C_1 \subseteq C_2$ and a real $\lambda \geq 1$, $C_1^{\lambda} \subseteq C_2^{\lambda}$.
\end{lemma}

The following lemma is a technical result, which shows that if a ray hits the interior of the base of a cap of width at least $\eps$, then it hits the interior of the base of a cap of width exactly $\eps$ that is contained in the original. 

\begin{lemma} \label{lem:cap-tech}
Let $0 < \eps < 1$, and let $K \subseteq \RE^n$ be a convex body containing the origin in its interior. Let $r$ be a ray shot from the origin, and let $D$ be a cap of $K$ of width at least $\eps$ such that ray $r$ intersects the interior of its base. Then there exists a cap $E \subseteq D$ of width $\eps$ such that ray $r$ intersects the interior of its base.
\end{lemma}

\begin{proof}
Let $p$ be the point of intersection of ray $r$ with the boundary of $K$. Let $F \subseteq D$ be the cap whose base passes through $p$ and is parallel to the base of $D$. We now consider two cases.

If the width of cap $F$ is less than $\eps$, then we let $E$ be the cap of width $\eps$ obtained by translating the base of $F$ parallel to itself (towards the base of $D$, as shown in Figure~\ref{f:subcap}(a)). Clearly $E \subseteq D$ and satisfies the conditions specified in the lemma.

\begin{figure}[htbp]
  \centerline{\includegraphics[scale=.8,page=1]{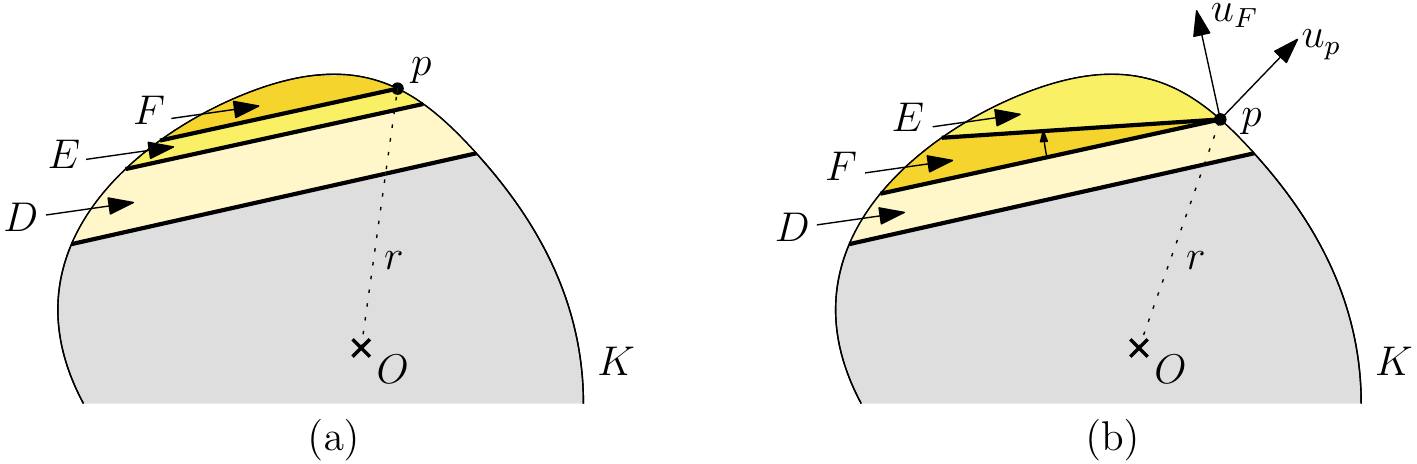}}
  \caption{\label{f:subcap}Proof of Lemma~\ref{lem:cap-tech}.}
\end{figure}

Otherwise, if the width of cap $F$ is at least $\eps$, then intuitively, we can rotate its base about $p$ (shrinking cap $F$ in the process), until its width is infinitesimally smaller than $\eps$ (Figure~\ref{f:subcap}(b)). More formally, let $u_F$ denote the normal vector for $F$'s base and let $u_p$ denote the  (any) surface normal vector to $K$ at $p$ (both unit length). Since $p$ is on the boundary, the cap orthogonal to $u_p$ and passing through $p$ has width zero. Since $F$ has width at least $\eps$, $u_F \neq u_p$. 

Considering the 2-dimensional linear subspace spanned by $u_F$ and $u_p$, we rotate continuously from $u_F$ to $u_p$, and consider the hyperplane passing through $p$ orthogonal to this vector. Clearly, the width of the associated cap varies continuously from $\width(F)$ to zero. Thus, there must be an angle where the cap width is infinitesimally smaller than $\eps$. We can expand this cap by translating its base parallel to itself to obtain a cap $E$ of width $\eps$, which satisfies all the conditions specified in the lemma.
\end{proof}

\subsection{Dual Caps and Cones} \label{s:dcaps}

It will be useful to consider the notion of a cap in a dual setting (see, e.g., \cite{AFM12b, AFM12a}). Given a convex body $K \subseteq \RE^n$ and a point $z$ that is exterior to $K$, we define the \emph{dual cap} of $K$ with respect to $z$, denoted $\dcap{K}{z}$, to be the set of $(n-1)$-dimensional hyperplanes that pass through $z$ and do not intersect $K$'s interior (see Figure~\ref{f:dual-cap-def}). In this paper, $K$ will be either full dimensional or one dimension less. We define the polar of a dual cap to be the set of points that results by taking the polar of each hyperplane of the dual cap. 

\begin{figure}[htbp]
  \centerline{\includegraphics[scale=.8]{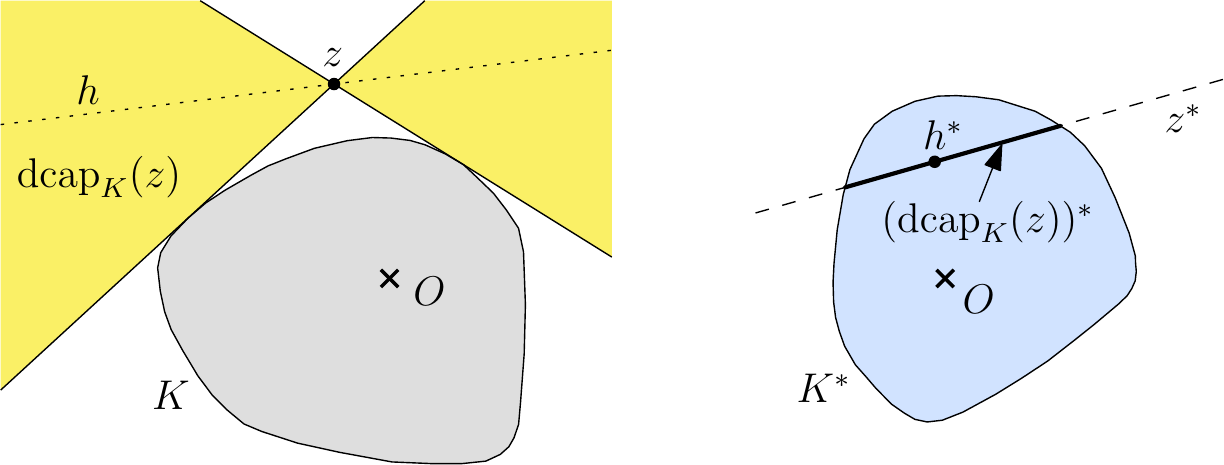}}
  \caption{\label{f:dual-cap-def}A dual cap and its polar.}
\end{figure}

Given $z$ exterior to $K$, and consider the cap of $K^*$ induced by the hyperplane $z^*$. By standard properties of the polar transformation, a hyperplane $h \in \dcap{K}{z}$ if and only if the point $h^*$ lies on $K^* \cap z^*$. As an immediate consequence, we obtain the following relationship between caps and dual caps.

\begin{lemma} \label{lem:polardcap}
Let $K \subseteq \RE^n$ be a full dimensional convex body that contains the origin and let $z \not\in K$. Then $(\dcap{K}{z})^* = \base(\pcap{K^*}{z^*})$.
\end{lemma}

Another useful concept involves cones induced by external points. A convex body $K$ and a point $z \not\in K$ naturally define two infinite convex cones. The \emph{inner cone}, denoted $\icone{K}{z}$, is the intersection of all the halfspaces that contain $K$ whose bounding hyperplanes pass through $z$ (see Figure~\ref{f:dualcaps2}). Equivalently, $\icone{K}{z}$ is the set of points $p$ such that the ray $z p$ intersects $K$. The \emph{outer cone}, denoted $\ocone{K}{z}$, is defined analogously as the intersection of halfspaces passing through $z$ that do not contain any point of $K$ (see Figure \ref{f:cone}). It is easy to see that $\ocone{K}{z}$ is the reflection of $\icone{K}{z}$ about $z$. The following lemma shows that membership in the outer cone and containment of caps are related through duality.

\begin{figure}[htbp]
\centering
\includegraphics[scale=.8]{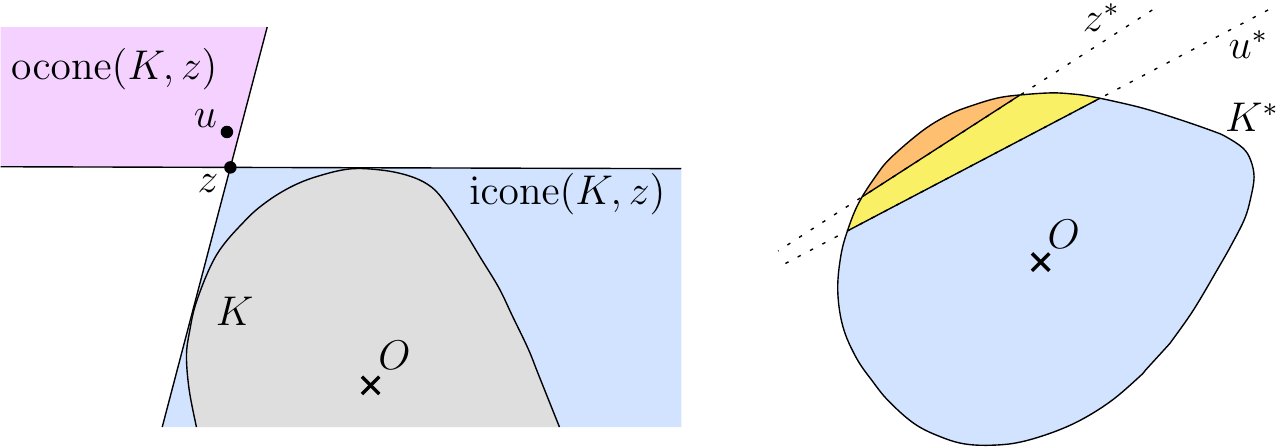}
\caption{\label{f:cone} Inner and outer cones.}
\end{figure}

\begin{lemma} \label{lem:ocone}
Let $K$ be a convex body with the origin $O$ in its interior. Then $u \in \ocone{K}{z}$ if and only if $\pcap{K^*}{z^*} \subseteq \pcap{K^*}{u^*}$.
\end{lemma}

\begin{proof}
By definition, $u \in \ocone{K}{z}$ if and only if any hyperplane $h$ that separates $z$ from $K$ also separates $u$ from $K$. Also, by standard properties of the polar transformation, a hyperplane $h$ separates $z$ from $K$ if and only if the point $h^* \in \pcap{K^*}{z^*}$. Similarly, hyperplane $h$ separates $u$ from $K$ if and only if the point $h^* \in \pcap{K^*}{u^*}$. Thus, the condition $u \in \ocone{K}{z}$ is equivalent to the condition $\pcap{K^*}{z^*} \subseteq \pcap{K^*}{u^*}$.
\end{proof}

\subsection{Macbeath Regions}

Given a convex body $K$ and a point $x \in K$, and a scaling factor $\lambda > 0$, the \emph{Macbeath region} $M_K^\lambda(x)$ is defined as
\[
    M_K^\lambda(x) ~= ~ x + \lambda ((K - x) \cap (x - K)).
\]
It is easy to see that $M_K^1(x)$ is the intersection of $K$ with the reflection of $K$ around $x$, and so $M_K^1(x)$ is centrally symmetric about $x$. Indeed, it is the largest centrally symmetric body centered at $x$ and contained in $K$. Furthermore, $M_K^\lambda(x)$ is a copy of $M_K^1(x)$ scaled by the factor $\lambda$ about the center $x$ (see the right side of Figure~\ref{f:mahlermac}). We will omit the subscript $K$ when the convex body is clear from the context. As a convenience, we define $M(x) = M^1(x)$.

We now present lemmas that encapsulate standard properties of Macbeath regions. The first lemma implies that a (shrunken) Macbeath region can act as a proxy for any other (shrunken) Macbeath region overlapping it~\cite{ELR70,BCP93}. Our version uses different parameters and is proved in~\cite{AFM17a} (Lemma~2.4).

\begin{lemma} \label{lem:mac-mac}
Let $K$ be a convex body and let $\lambda \le \frac{1}{5}$ be any real. If $x, y \in K$ such that $M^{\lambda}(x) \cap M^{\lambda}(y) \neq \emptyset$, then $M^{\lambda}(y) \subseteq M^{4\lambda}(x)$.
\end{lemma}

The following lemmas are useful in situations when we know that a Macbeath region overlaps a cap of $K$, and allow us to conclude that a constant factor expansion of the cap will fully contain the Macbeath region. The first applies to shrunken Macbeath regions and the second to Macbeath regions with any scaling factor. The proof of the first appears in~\cite{AFM17c} (Lemma~2.5), and the second is an immediate consequence of the definition of Macbeath regions.

\begin{lemma} \label{lem:mac-cap}
Let $K$ be a convex body. Let $C$ be a cap of $K$ and $x$ be a point in $K$ such that $C \cap M^{1/5}(x) \neq \emptyset$. Then $M^{1/5}(x) \subseteq C^2$.
\end{lemma}

\begin{lemma} \label{lem:mac-cap-var}
Let $K$ be a convex body and $\lambda > 0$. If $x$ is a point in a cap $C$ of $K$, then $M^\lambda(x) \cap K \subseteq C^{1+\lambda}$.
\end{lemma}

Points in a shrunken Macbeath region are similar in many respects. For example, they have similar ray distances.

\begin{lemma}
\label{lem:core-ray}
Let $K$ be a convex body. If $x$ is a $\frac{1}{2}$-shallow point in $K$ and $y \in M^{1/5}(x)$, then $\ray(x)/2 \leq \ray(y) \leq 2 \ray(x)$.
\end{lemma}

\begin{proof}
Let $C_x$ denote the minimum width cap for $x$. By Lemma~\ref{lem:min-width-cap}, $\width(C_x) = \ray(x)$. Also, by Lemma~\ref{lem:mac-cap}, we have $M^{1/5}(x) \subseteq C_x^2$ and so $y \in C_x^2$. It follows from Lemma~\ref{lem:raydist-width} that $\ray(y) \leq \width(C_x^2) = 2 \width(C_x)$. Thus $\ray(y) \leq 2 \ray(x)$, which proves the second inequality. To prove the first inequality, note that this follows trivially unless $\ray(y) \leq \frac{1}{4}$ (since $\ray(x) \leq \frac{1}{2}$). If $\ray(y) \leq \frac{1}{4}$, consider the minimum width cap $C_y$ for $y$. By Lemma~\ref{lem:min-width-cap}, $\width(C_y) = \ray(y)$. Also, by Lemma~\ref{lem:mac-cap}, we have $M^{1/5}(x) \subseteq C_y^2$ and so $x \in C_y^2$. It follows from Lemma~\ref{lem:raydist-width} that $\ray(x) \leq \width(C_y^2) = 2 \width(C_y)$. Thus $\ray(x) \leq 2 \ray(y)$, which completes the proof.
\end{proof}

The remaining lemmas in this section relate caps with the associated Macbeath regions.

\begin{lemma}[B{\'a}r{\'a}ny~\cite{Bar07}] \label{lem:min-vol-cap1}
Given a convex body $K \subseteq \RE^n$, let $C$ be a $\frac{1}{3}$-shallow cap of $K$, and let $p$ be the centroid of $\base(C)$. Then $C \subseteq M^{2n}(p)$.
\end{lemma}

\begin{lemma}
\label{lem:wide-cap}
Let $0 < \beta < 1$ be any constant. Let $K \subseteq \RE^n$ be a well-centered convex body, $p \in K$, and $C$ be the minimum volume cap associated with $p$. If $C$ contains the origin or $\width(C) \geq \beta$, then $\vol_K(M(p)) \geq 2^{-O(n)}$.
\end{lemma}

\begin{proof}
We claim that $K$ satisfies the Winternitz property with respect to $p$. Note this is equivalent to the claim that $\vol_K(C) \geq 2^{-O(n)}$. 

We consider two cases. First, suppose that $C$ contains the origin. Since $K$ is well-centered, by Lemma~\ref{lem:centroid}, $K$ satisfies the Winternitz property with respect to the origin. It follows that $\vol_K(C) \geq 2^{-O(n)}$. Otherwise, if $C$ does not contain the origin, then since the width of $C$ is at least $\beta$, the expanded cap $C^{1/\beta}$ contains the origin. By Lemma~\ref{lem:cap-exp}, $\vol(C^{1/\beta}) \leq 2^{O(n)} \vol(C)$. Again, using the fact that $K$ satisfies the Winternitz property with respect to the origin, we have $\vol_K(C^{1/\beta}) \geq 2^{-O(n)}$. Thus, in both cases, $\vol_K(C) \geq 2^{-O(n)}$, which proves the claim. 

Since $K$ satisfies the Winternitz property with respect to $p$, by Lemma~\ref{lem:centroid}, it must satisfy the Kovner-Besicovitch property with respect to $p$. Thus $\vol_K(M(p)) = \vol_K((K-p) \cap (p-K)) \geq 2^{-O(n)}$, as desired.
\end{proof}

\begin{lemma}
\label{lem:min-vol-cap2}
Given a convex body $K \subseteq \RE^n$, let $C$ be a $\frac{1}{3}$-shallow cap of $K$, and let $p$ be the centroid of $\base(C)$. We have
\[
    2^{-O(n)} \cdot \vol(C) 
        ~ \leq ~ \vol(M(p)) 
        ~ \leq ~ 2 \cdot \vol(C).
\]
\end{lemma}

\begin{proof}
The second inequality holds easily because half of $M(p)$ lies inside $C$. To prove the first inequality, let $B = \base(C)$, let $a = \area(B)$ denote its $(n-1)$-dimensional volume, and let $B' = M(p) \cap B$. Treating $p$ as the origin of the coordinate system, by definition of Macbeath regions, $B' = B \cap - B$. By applying Lemma~\ref{lem:centroid} (to the hyperplane containing $B$) we have $\area(B') \geq a / 2^{O(n)}$.

Let $x$ denote the apex of $C$, and let $x'$ be the farthest point on segment $\overline{p x}$ that is contained in $M(p)$. By Lemma~\ref{lem:min-vol-cap1}, $\|p x'\| \geq \|p x\| / 2 n$. By convexity, the generalized cone $P = \conv(B' \cup \{x'\})$ is contained within $M(p)$. Letting $w$ denote the absolute width of $C$, the height of this cone is at least $w/2 n$. Thus 
\[
    \vol(M(p)) 
        ~ \geq ~ \vol(P) 
        ~ \geq ~ \frac{\area(B') \cdot w/2 n}{n} 
        ~ \geq ~ \frac{(a / 2^{O(n)}) \cdot w/2 n}{n} 
        ~ =    ~ \frac{a w}{ n^2 2^{O(n)}}.
\]
By Lemma~\ref{lem:vol-cap}, $\vol(C) \leq 2^{n-1} a w$, and thus, 
\[
    \vol(M(p)) 
        ~ \geq ~ 2^{-O(n)} \cdot \vol(C),
\]
as desired.
\end{proof}

\begin{corollary}
\label{cor:min-vol-cap2}
Let $K \subseteq \RE^n$ be a convex body, $p \in K$, and $C$ be the minimum volume cap associated with $p$. We have
\[
    2^{-O(n)} \cdot \vol(C) 
        ~ \leq ~ \vol(M(p)) 
        ~ \leq ~ 2 \cdot \vol(C).
\]
\end{corollary}

\begin{proof}
The second inequality holds for the same reason as in Lemma~\ref{lem:min-vol-cap2}. To prove the first inequality, recall the well-known property of minimum volume caps that $p$ is the centroid of the base of its associated minimum volume cap~\cite{ELR70}. Treating the centroid of $K$ as the origin, we consider two cases. If $C$ is $(1/3)$-shallow, then the corollary follows from Lemma~\ref{lem:min-vol-cap2}. Otherwise, $C$ contains the origin or its width is at least $1/3$. Noting that $K$ is well-centered with respect to the centroid (Lemma~\ref{lem:centroid}) and applying Lemma~\ref{lem:wide-cap}, it follows that $\vol_K(M(p)) \ge 2^{-O(n)}$. That is, $\vol(M(p)) \geq 2^{-O(n)} \vol(K) \geq 2^{-O(n)} \vol(C)$, which completes the proof.
\end{proof}

\subsection{Similar Caps} 
\label{s:similar}

The Macbeath regions of a convex body $K$, and more specifically, its shrunken Macbeath regions, provide an affine-invariant notion of the closeness between points, through the property that both points lie within the same shrunken Macbeath region. We would like to define a similar affine-invariant notion of closeness between caps. We say that two caps $C_1$ and $C_2$ are \emph{$\lambda$-similar} for $\lambda \ge 1$, if $C_1 \subseteq C_2^{\lambda}$ and $C_2 \subseteq C_1^{\lambda}$ (see Figure~\ref{f:similar}(a)). If two caps are $\lambda$-similar for a constant $\lambda$, we say that the caps are \emph{similar}. 

\begin{figure}[htbp]
\centering
\includegraphics[scale=.8]{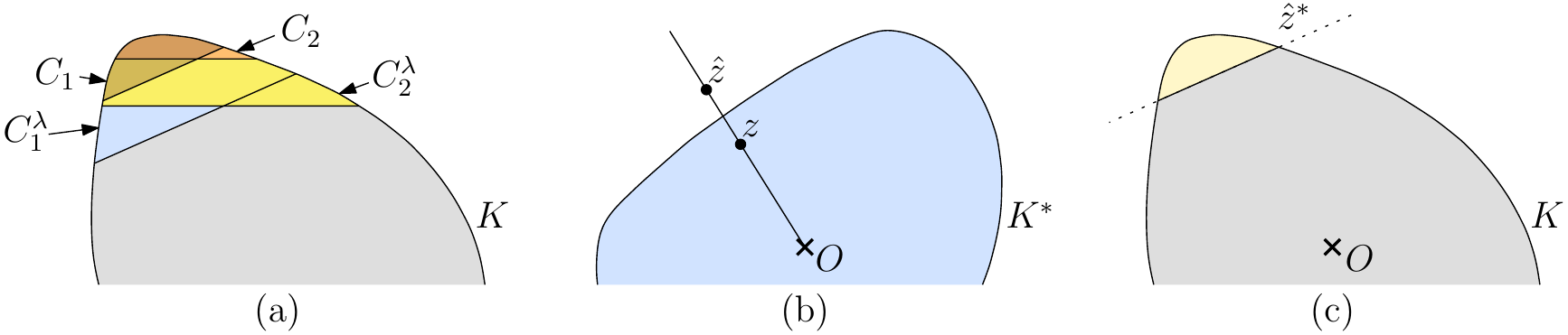}
\caption{\label{f:similar} Similar caps and $\eps$-representative caps.}
\end{figure}

It is natural to conjecture that these two notions of similarity are related through duality. In order to establish such a relationship consider the following mapping. Consider a point $z \in K^*$. Take a point $\hat{z} \not\in K^*$ on the ray $O z$ such that $\ray(\hat{z}) = \eps$ (see Figure~\ref{f:similar}(b)). The dual hyperplane $\hat{z}^*$ intersects $K$, and so induces a cap, which we call $z$'s \emph{$\eps$-representative cap} (see Figure~\ref{f:similar}(c)). The main result of this section is Lemma~\ref{lem:sandwich}, which shows that points lying within the same shrunken Macbeath region have similar representative caps. Before proving this, we begin with a technical lemma.

\begin{lemma} \label{lem:guarding}
Let $\alpha \le \frac{1}{8}$. Let $y \in K^*$ be an $\alpha$-shallow point. Consider two rays $r$ and $r'$ shot from the origin through $M^{1/5}(y)$ (see Figure \ref{f:guarding}). Let $z \not\in K^*$ be an $\alpha$-shallow point on $r$ and let $u \not\in K^*$ be a point on $r'$ such that $\ray(u) > 4 \ray(y) + 2 \ray(z)$. Then $\pcap{K}{z^*} \subseteq \pcap{K}{u^*}$.
\end{lemma}

\begin{figure}[htbp]
\centering
\includegraphics[scale=.8]{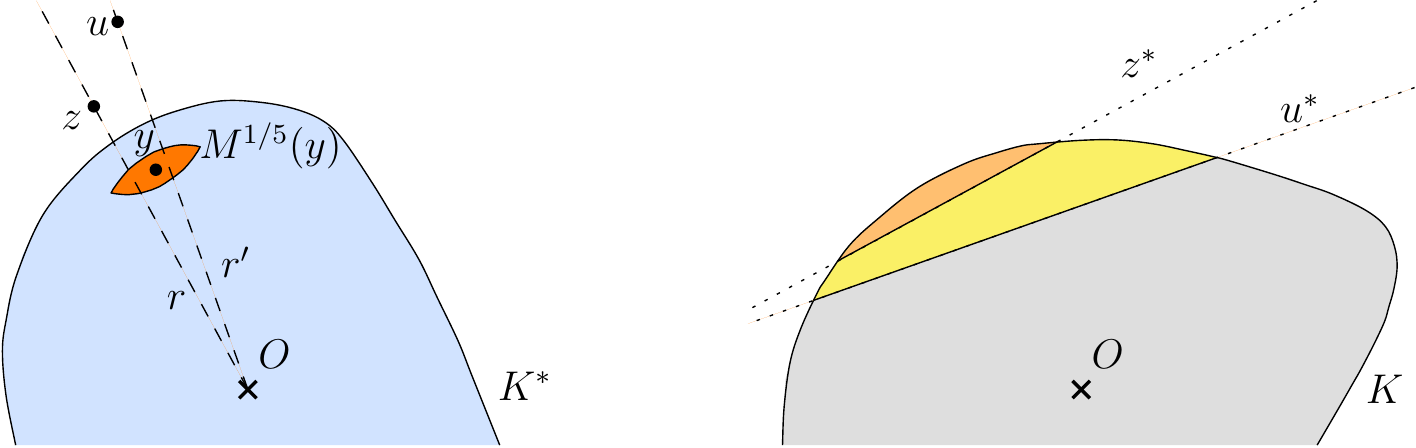}
\caption{\label{f:guarding} Statement of Lemma~\ref{lem:guarding}. 
}
\end{figure}

\begin{proof}
Let $h$ be any hyperplane passing through $z$ that does not intersect $K^*$. We will show that $h$ separates $u$ from $K^*$. This would imply that $u \in \ocone{K^*}{z}$, and the result would then follow from Lemma~\ref{lem:ocone}.

Let $p$ be any point in $r \cap M^{1/5}(y)$. By Lemma~\ref{lem:core-ray}, we have $\ray(p) \le 2 \ray(y)$. Consider a hyperplane $h'$ that is parallel to $h$ and passes through $p$ (see Figure \ref{f:guardingproof}). Let $C$ be the cap induced by $h'$. Letting $t$ denote the point of intersection of ray $r$ with $\bd K^*$, we have
\begin{equation} \label{eq:guarding}
  \width(C) 
    ~ \leq ~ \frac{\|p z\|}{\|O z\|} 
    ~ =    ~ \frac{\|p t\| + \|t z\|}{\|O z\|} 
    ~ \leq ~ \frac{\|p t\|}{\|O t\|} + \frac{\|t z\|}{\|O z\|} 
    ~ =    ~ \ray(p) + \ray(z) 
    ~ \leq ~ 2 \ray(y) + \ray(z).
\end{equation}

Since $C$ intersects $M^{1/5}(y)$, by Lemma~\ref{lem:mac-cap}, the cap $C^2$ encloses $M^{1/5}(y)$. Since $y$ and $z$ are $\alpha$-shallow for $\alpha = \frac{1}{8}$, by Eq.~\eqref{eq:guarding} we have $\width(C) \le 3/8$. It follows $\width(C^2) < 1$, and hence $O$ lies outside $C^2$. Let $h''$ denote the hyperplane passing through the base of $C^2$. Since $r'$ intersects $M^{1/5}(y)$, it follows that $r'$ must intersect $h''$ and $h$. Let $z'$ denote the point of intersection of $r'$ with $h$. We will show that $\ray(z') \le 4 \ray(y) + 2 \ray(z)$. Recalling from the statement of the lemma that $\ray(u) > 4 \ray(y) + 2 \ray(z)$, this would imply that $h$ separates $u$ from $K^*$, as desired.

\begin{figure}[htbp]
\centering
\includegraphics[scale=.8]{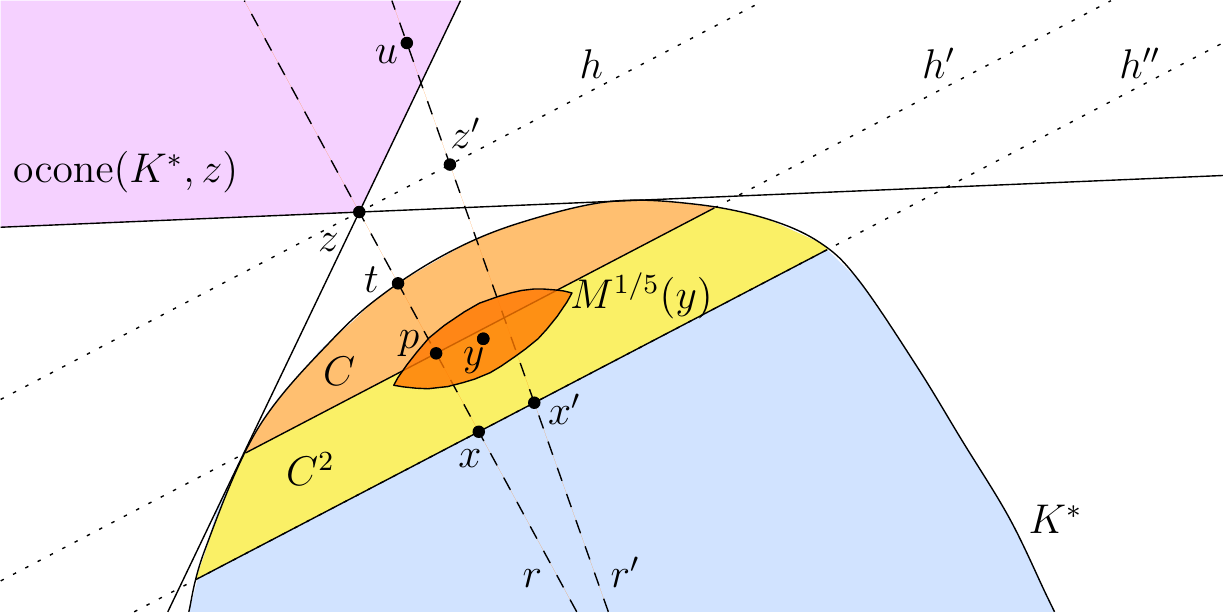}
\caption{\label{f:guardingproof} Proof of Lemma~\ref{lem:guarding}. }
\end{figure}

Let $x$ and $x'$ denote the points of intersection of the rays $r$ and $r'$, respectively, with $h''$. By similar triangles we have $\ray(z') \le \|x' z'\| / \|O z'\| = \|x z\| / \|O z\|$. Observe that the distance between $h''$ and $h'$ is no more than the distance between $h'$ and $h$, and so $\|x z\| \le 2 \|p z\|$. Combining this with Eq.~\eqref{eq:guarding}, we obtain
\[
    \ray(z') 
        ~ \leq ~ \frac{\|x z\|}{\|O z\|} 
        ~ \leq ~ \frac{2 \|p z\|}{\|O z\|} 
        ~ \leq ~ 2 (2 \ray(y) + \ray(z)) 
        ~ =    ~ 4 \ray(y) + 2 \ray(z),
\]
which completes the proof.
\end{proof}

We now establish the main result of this section.

\begin{lemma}
\label{lem:sandwich} 
Let $\eps \leq \frac{1}{16}$, and let $y \in K^*$ such that $\ray(y) \leq \eps$. For any two points $x, z \in M^{1/5}(y)$, their respective $\eps$-representative caps are 8-similar.
\end{lemma}

\begin{proof} 
Let $x_1$ and $z_1$ be points external to $K^*$ both at ray distance $\eps$ on the rays $O x$ and $O z$, respectively (see Figure \ref{f:sandwichproof}(a)). Let $C_x$ and $C_z$ denote the $\eps$-representative caps of $x$ and $z$, respectively (see Figure \ref{f:sandwichproof}(b)). Recall that $C_x$ and $C_z$ are the caps in $K$ induced by $x_1^*$ and $z_1^*$, respectively. By standard properties of the polar transformation $\width(C_x) = \ray(x_1) = \eps$, and similarly, $\width(C_z) = \ray(z_1) = \eps$. Let $x_2$ and $z_2$ be points external to $K^*$ both at ray distance $8 \eps$ on the rays $O x$ and $O z$, respectively (see Figure~\ref{f:sandwichproof}). By our bound on $\eps$, these ray distances are at most $\frac{1}{2}$. Clearly, $x_2^*$ and $z_2^*$ induce the caps $C_x^8$ and $C_z^8$ in $K$, respectively.

\begin{figure}[htbp]
\centering
\includegraphics[scale=.8,page=1]{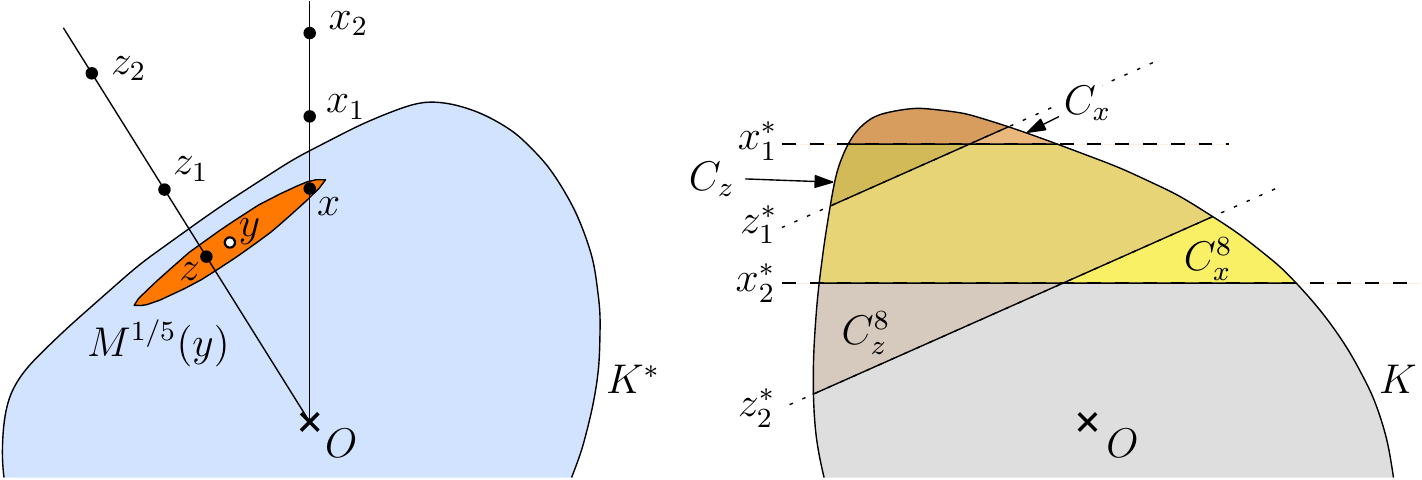}
\caption{\label{f:sandwichproof} Proof of Lemma~\ref{lem:sandwich}.}
\end{figure}

Since $\ray(x_2) = 8\eps, \ray(y) \leq \eps$ and $\ray(z_1) = \eps$, we have $\ray(x_2) > 2\ray(z_1) + 4\ray(y)$. It follows from Lemma~\ref{lem:guarding} that $C_z \subseteq C_x^8$. A symmetrical argument shows that $C_x \subseteq C_z^8$. Therefore $C_x$ and $C_z$ are 8-similar, as desired.
\end{proof}

The next lemma shows that similarity holds, even if ray distances are altered by a constant factor.

\begin{corollary} \label{cor:sandwich} 
Let $\eps \leq \frac{1}{16}$, and let $y \in K^*$ such that $\ray(y) \leq \eps$. Let $C_x$ be a cap of $K$ such that $\eps/2 \leq \width(C_x) \leq 2\eps$, and such that the ray shot from the origin orthogonal to the base of $C_x$ intersects $M^{1/5}(y)$. Then the cap $C_x$ and the $\eps$-representative cap $C_z$ of any point $z \in M^{1/5}(y)$ are 16-similar.
\end{corollary}

\begin{proof}
Let $r$ denote the ray shot from the origin orthogonal to the base of $C_x$. Let $x$ be any point that lies in $r \cap M^{1/5}(y)$. Let $C_x'$ be the $\eps$-representative cap of $x$. By Lemma~\ref{lem:sandwich}, the caps $C_x'$ and $C_z$ are 8-similar. Also, it follows from our choice of point $x$ that the caps $C_x$ and $C_x'$ have parallel bases and their widths differ by a factor of at most two. Thus $C_x$ and $C_x'$ are 2-similar. Using the fact that $C_x'$ and $C_z$ are 8-similar, and applying Lemma~\ref{lem:cap-containment-exp}, it is easy to see that $C_x$ and $C_z$ are 16-similar.
\end{proof}

\section{Caps in the Polar: Mahler Relationship} \label{s:mahler}

As mentioned in Section~\ref{s:techniques}, a central element of our analysis is establishing a Mahler-like reciprocal relationship between volumes of caps in $K$ and corresponding caps of $K^*$. While our new result is similar in spirit to those given by Arya {\etal}~\cite{AAFM22} and that of Nasz{\'o}di {\etal}~\cite{NNR20}, it is stronger than both. Compared to \cite{AAFM22}, the dependency of the Mahler volume on dimension is improved from $2^{-O(n\log n)}$ to $2^{-O(n)}$, which is critical in the high-dimensional setting in reducing terms of the form $n^{O(n)}$ to $2^{O(n)}$. Further, our result is presented in a cleaner form, which is affine-invariant. Compared to Nasz{\'o}di {\etal}~\cite{NNR20}, which was focused on sampling from just the boundary of $K$, our results can be applied to caps of varying widths, and hence it applies to sampling from the interior of $K$. This fact too is critical in the applications we consider. Our improvements are obtained by a more sophisticated geometric analysis and our affine-invariant approach.

For the sake of concreteness, we state the lemmas of this section in terms of an arbitrary direction, which we call ``vertical,'' and any hyperplane orthogonal to this direction is called ``horizontal.'' Since the direction is arbitrary, there is no loss of generality.

\subsection{Dual Caps and the Difference Body} \label{s:diff-body}

This subsection is devoted to a key construction in our analysis. Given a full dimensional convex body $K$ and a point $z \not\in K$, the following lemma identifies an $(n-1)$-dimensional body $\Upsilon$ such that $\dcap{\Upsilon}{z} = \dcap{K}{z}$, where $\Upsilon$ is related to the base $B$ of a certain $\eps$-width cap in the sense that $\Upsilon$ can be sandwiched between $B$ and a scaled copy of the difference body of $B$.

\begin{figure}[htb]
  \centerline{\includegraphics[scale=.8,page=2]{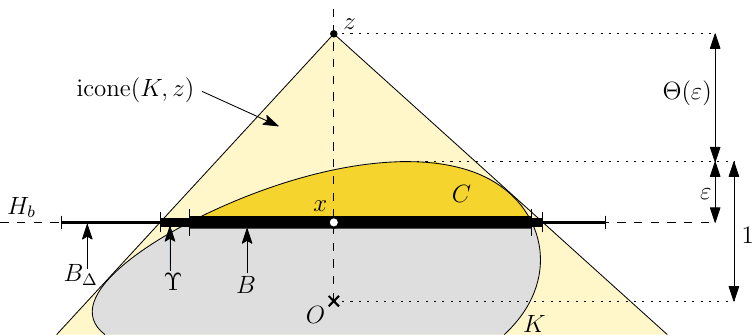}}
  \caption{\label{f:dualcaps2}Statement of Lemma~\ref{lem:sandwich-dualcaps}.}
\end{figure}

\begin{lemma} \label{lem:sandwich-dualcaps} 
Let $\eps \le \frac{1}{8}$. Let $K$ be a convex body with the origin $O$ in its interior. Let $z \notin K$ be a point on the ray from the origin directed vertically upwards such that $\ray(z) = 2\eps$. Consider an $\eps$-width cap $C$ above the origin whose base $B$ intersects $Oz$ and is horizontal. Let $H_b$ be the hyperplane passing through the base $B$, and let $\Upsilon = \icone{K}{z} \cap H_b$. Let $x$ denote the point of intersection of $B$ with $Oz$, and let $B_{\Delta} = 5\Delta(B) + x$. Then $B \subseteq \Upsilon \subseteq B_{\Delta}$ (see Figure~\ref{f:dualcaps2}). 
\end{lemma}

\begin{proof}
By definition, $K \subseteq \icone{K}{z}$, and so $B \subseteq \Upsilon$. Thus, it suffices to show that $\Upsilon \subseteq B_{\Delta}$. To prove this, we will show that $K \subseteq \icone{B_{\Delta}}{z}$.

Let $a$ denote an apex of $C$ and let $a'$ be the point obtained by projecting $a$ orthogonally onto $O z$ (see Figure~\ref{f:dualcaps3}). Without loss of generality, assume that $\|O a'\| = 1 $. Note that $\|xa'\| = \eps$, where $x$ is the point of intersection of the ray $O z$ with the base of cap $C$. It is easy to check that $\eps \le \|a'z\| \le 3 \eps$.

\begin{figure}[htb]
  \centerline{\includegraphics[scale=.8,page=1]{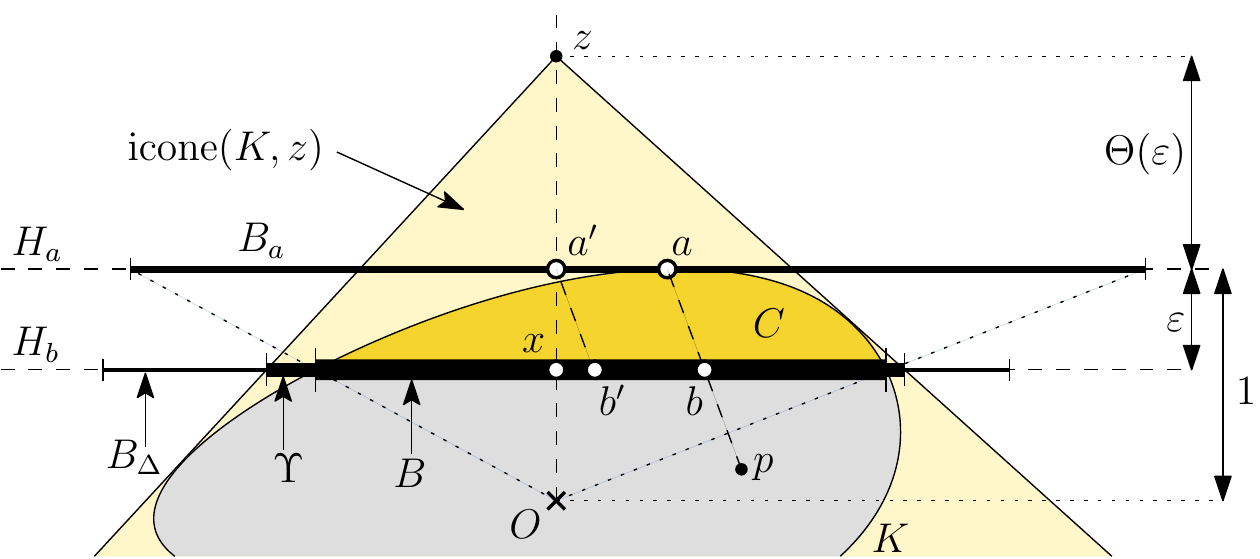}}
  \caption{\label{f:dualcaps3}Proof of Lemma~\ref{lem:sandwich-dualcaps}.}
\end{figure}

For the remainder of this proof, it will be convenient to imagine that the origin is at $x$. Our strategy will be to show that $C \subseteq \icone{2(1+2\eps)B}{z}$ and $K \setminus C \subseteq \icone{4(1+2\eps) \Delta(B)}{z}$. Since $B$ contains the origin, it follows easily that $B \subseteq \Delta(B)$. This implies that $K \subseteq \icone{4(1+2\eps) \Delta(B)}{z} \subseteq \icone{5\Delta(B)}{z}$ since $\eps \le \frac{1}{8}$. By definition of $B_{\Delta}$, this would complete the proof.

First, we will prove that $C \subseteq \icone{2(1+2\eps)B}{z}$. It follows from convexity that $C$ is contained in the truncated portion of $\icone{B}{O}$ between the hyperplane $H_b$ and the hyperplane above $H_b$ that is parallel to it at distance $\eps$ (call it $H_a$). Note that $\icone{B}{O} \cap H_a$ is the $(n-1)$-dimensional convex body obtained by scaling $B$ about $x$ by a factor of $1/(1-\eps)$ and translating it vertically upwards by amount $\eps$. Call 
this body $B_a$. (Formally, $B_a = (1/(1-\eps)) B + a'$.) It 
is easy to see that $C \subseteq \icone{B_a}{z}$. Since $\|z x\| \le 2 \|z a'\|$, it follows that $\icone{B_a}{z} \cap H_b \subseteq 2 (1/(1-\eps)) B$. Thus $C \subseteq \icone{2 (1/(1-\eps)) B}{z} \subseteq \icone{2 (1+2\eps) B}{z}$, where in the last containment we have used the fact that $\eps \le \frac{1}{8}$.

It remains to prove that $K \setminus C \subseteq \icone{4(1+2\eps) \Delta(B)}{z}$. By convexity, it follows that $K \setminus C \subseteq \icone{B}{a}$. Define $t = a' - a$ and $B^+ = \conv(B \cup (B+t))$. We claim that $K \setminus C \subseteq \icone{B^+}{a'}$. To prove this, let $p$ be any point in $K \setminus C$. Since $K \setminus C \subseteq \icone{B}{a}$, it follows that $\overline{a p}$ intersects the base $B$; let $b$ denote this point of intersection. Since $b \in B$, we have $b \in B^+$. Define $b' = b + t$. Clearly $b' \in B + t$ and hence $b' \in B^+$. Note that the points $b, b', a', a$ form a parallelogram (because $b'-a' = b-a$). By elementary geometry, $p$ also lies in the 2-dimensional flat of this parallelogram and $\overline{a' p}$ intersects $\overline{b b'}$. Since $b, b' \in B^+$ and $B^+$ is convex, it follows that $\overline{b b'}$ is contained in $B^+$. Thus $\overline{a' p}$ intersects $B^+$, which implies that $p \in \icone{B^+}{a'}$. This proves that $K \setminus C \subseteq \icone{B^+}{a'}$, as desired.

Next consider the cone obtained by translating $\icone{B^+}{a'}$ vertically upwards to $z$. Clearly the resulting cone contains $K \setminus C$, and since $\|z x\| \le 4\|a' x\|$, it follows that the intersection of this cone with $H_b$ is contained in $4B^+$. Thus $K \setminus C \subseteq \icone{4B^+}{z}$.

To complete the proof we need to relate $B^+$ to $\Delta(B)$. To be precise, we will show that $B^+ \subseteq (1+2\eps) \Delta(B)$. Recall that $B^+ = \conv(B \cup (B+t))$. By our earlier remarks, $a \in B_a$ and hence $-t = a - a' \in (1/(1-\eps)) B$. It follows that $B + t \subseteq (1/(1-\eps)) B - (-t) \subseteq \Delta((1/(1-\eps)) B)$, where the first containment is trivial and the second is immediate from the definition of difference bodies. Also, $B \subseteq \Delta((1/(1-\eps)) B)$ holds trivially. By convexity of difference bodies, it follows that $B^+ \subseteq \Delta((1/(1-\eps)) B)$. Thus $B^+ \subseteq (1/(1-\eps)) \Delta(B) \subseteq (1+2\eps) \Delta(B)$. Recalling that $K \setminus C \subseteq \icone{4B^+}{z}$, it follows that $K \setminus C \subseteq \icone{4 (1+2\eps) \Delta(B)}{z}$, which completes the proof.
\end{proof}

\subsection{Relating Caps in the Primal and Polar} \label{s:primal-polar}

In order to establish a Mahler-like relation between the volumes of caps of $K$ and $K^*$, it will be helpful to consider projections in one lower dimension, $n-1$. We will make use of a special case of a result appearing in \cite{AAFM22} (Lemma~{3.1}).
Consider a convex body $K$ lying on an $(n-1)$-dimensional hyperplane and a point $z$ that lies on the opposite side of this hyperplane from the origin (see Figure~\ref{f:polars}). The polar of the dual cap of $K$ with respect to $z$ is an $(n-1)$-dimensional convex body on the hyperplane $z^*$. Letting $G$ denote this object, the following lemma shows that if we project both $K$ and $G$ onto a suitable $(n-1)$-dimensional hyperplane, $G$ is the polar of $K$ up to scale factor.  

\begin{lemma}[Arya {\etal}~\cite{AAFM22}] \label{lem:polars}
Let $z \in \RE^n$ be a point that lies on a vertical ray from the origin $O$, and let $K$ be an $(n-1)$-dimensional convex body whose interior intersects the segment $O z$ at some point $x$. Further, suppose that $K$ lies on a hyperplane orthogonal to $Oz$. Let $G = (\dcap{K}{z})^*$ and let $t$ be the point of intersection of the vertical ray from $O$ with $z^*$. Then $G - t = \alpha (K - x)^*$, where $\alpha = \|x z\| / \|O z\|$.
\end{lemma}

\begin{figure}[htbp]
 \centerline{\includegraphics[scale=.8]{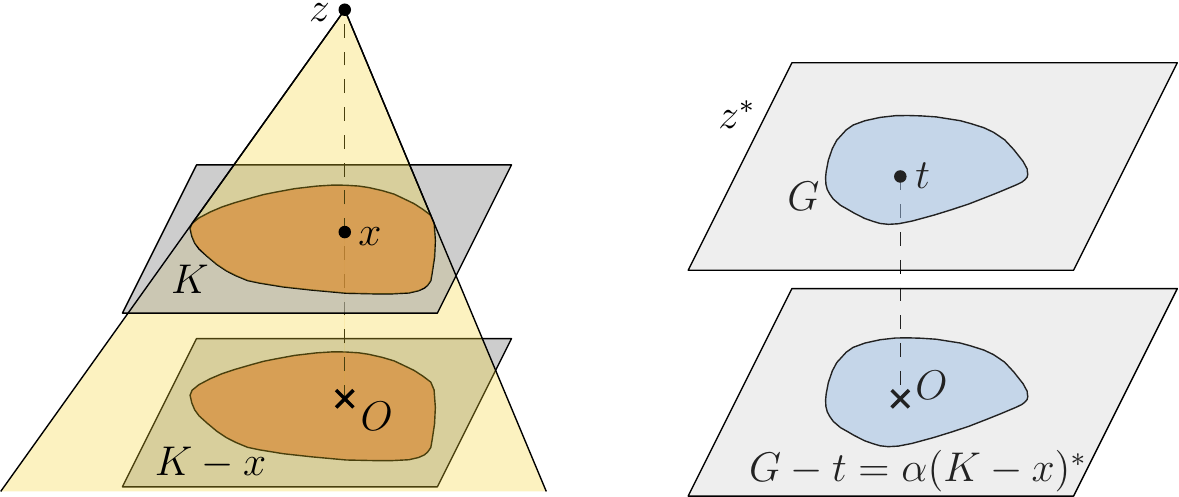}}
 \caption{\label{f:polars}Statement of Lemma~\ref{lem:polars}.}
\end{figure}

The following lemma describes the correspondence between caps in $K$ and its polar $K^*$, and it establishes the critical Mahler-type relationship between the volumes of these caps.

\begin{lemma} \label{lem:vol-product}
Let $0 < \eps \leq \frac{1}{8}$, and let $K \subseteq \RE^n$ be a well-centered convex body. Let $C$ be a cap of $K$ of width at least $\eps$. Consider the ray shot from the origin orthogonal to the base of $C$, and let $D$ be a cap of $K^*$ of width at least $\eps$ such that this ray intersects the interior of its base (see Figure~\ref{f:volprod}).  Then
\[
    \vol_K(C) \cdot \vol_{K^*}(D) 
        ~ \geq ~ 2^{-O(n)} \eps^{n+1}.
\]
\end{lemma}

\begin{figure}[htbp]
  \centerline{\includegraphics[scale=.8,page=1]{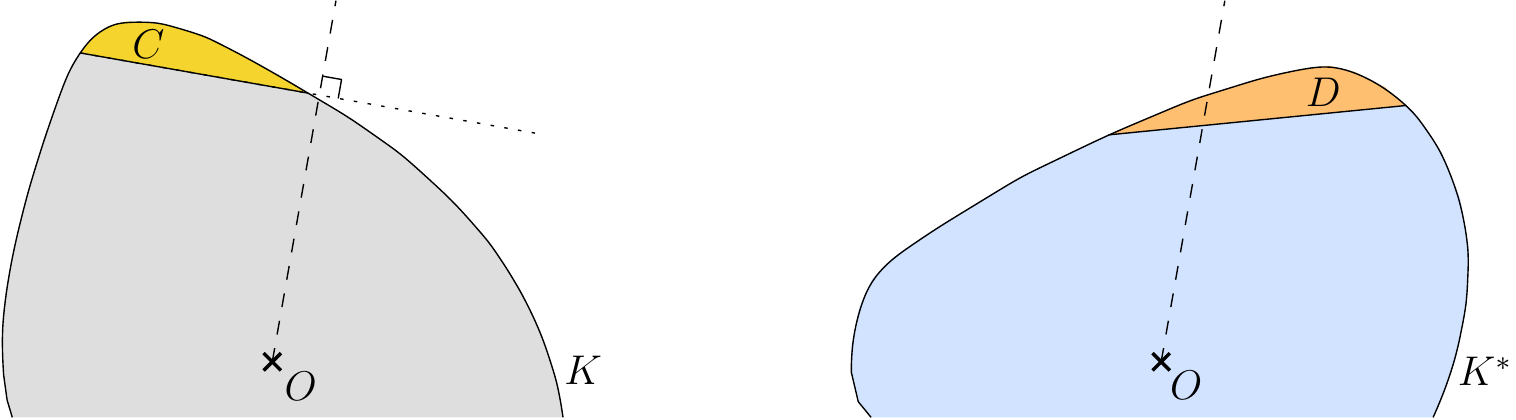}}
  \caption{\label{f:volprod}Statement of Lemma~\ref{lem:vol-product}.}
\end{figure}

\begin{proof}
Let $C'$ be a cap of width $2\eps$ whose base is parallel to the base of $C$ and which is on the same side of the origin as $C$. Clearly such a cap can be obtained by translating the base of $C$ parallel to itself. Note that $C' \subseteq C^2$ and so, by Lemma~\ref{lem:cap-exp}, it follows that $\vol(C') \le 2^{O(n)} \cdot \vol(C)$. Let $r$ denote the ray in the polar space, emanating from the origin of $K^*$ in a direction orthogonal to the base of $C$ (see Figure~\ref{f:volprodproof}). Recall that $r$ intersects the interior of the base of $D$. By Lemma~\ref{lem:cap-tech}, we can find a cap $D' \subseteq D$ whose width is $\eps$ and such that ray $r$ intersects the interior of the base of $D'$.  It is now easy to see that it suffices to prove the lemma with $C'$ and $D'$ in place of $C$ and $D$, respectively. As a convenience, in the remainder of this proof, we will write $C$ and $D$ in place of $C'$ and $D'$, respectively.

\begin{figure}[htbp]
  \centerline{\includegraphics[scale=.8,page=2]{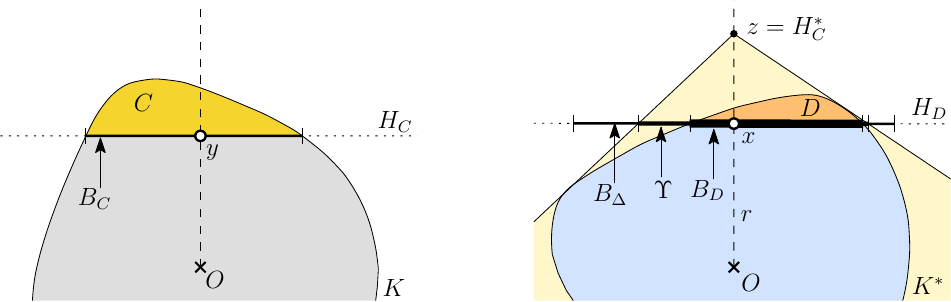}}
  \caption{\label{f:volprodproof}Proof of Lemma~\ref{lem:vol-product}.}
\end{figure}

As the product considered in this lemma is affine-invariant, we will apply a suitable linear transformation to simplify the subsequent analysis. Specifically, we apply a linear transformation in the polar space such that the base of $D$ becomes horizontal while the ray $r$ is directed vertically upwards. It is easy to see that the effect of this transformation in the original space is to make the base of cap $C$ horizontal (because it is the polar of a point on ray $r$). To summarize, after the transformation, the hyperplanes passing through the bases of the caps $C$ and $D$ are horizontal and above the origin and as relative measures the widths of both caps are unchanged. Further, the ray $r$ is directed vertically upwards in the polar and intersects the interior of the base of $D$. Also, after uniform scaling, we may assume that the absolute distance between the origin and the supporting hyperplane of cap $C$ that is parallel to its base is unity.

Let $B_C$ denote the base of cap $C$ and $H_C$ denote the hyperplane passing through $B_C$. Also, let $B_D$ denote the base of cap $D$ and $H_D$ denote the hyperplane passing through $B_D$. Define $z = H_C^*$. Note that $z$ lies outside $K^*$ on the ray from the origin directed vertically upwards and $\ray(z) = \width(C) = 2\eps$. By Lemma~\ref{lem:polardcap}, $B_C = (\dcap{K^*}{z})^*$. Define $\Upsilon = \icone{K^*}{z} \cap H_D$. Clearly $\dcap{K^*}{z} = \dcap{\Upsilon}{z}$. Thus $B_C = (\dcap{\Upsilon}{z})^*$.

Let $y$ denote the point of intersection of the vertical ray from $O$ with $B_C$, and let $x$ denote the point of intersection of the vertical ray from $O$ with $B_D$. Henceforth, in this proof, we will treat $y$ as the origin in the primal space and $x$ as the origin in the polar space. Applying Lemma~\ref{lem:polars} (setting $K$ in that lemma to $\Upsilon$), it follows that $B_C = \alpha \Upsilon^*$, where $\alpha = \|x z\| / \|O z\|$. Noting that $B_C$ is $(n-1)$-dimensional and $\alpha = \Theta(\eps)$, it follows that
\[
    \area(B_C) ~ \geq ~ 2^{-O(n)} \eps^{n-1} \cdot \area(\Upsilon^*).
\]
By Lemma~\ref{lem:vol-cap}, we have $\vol(C) \geq 2^{-O(n)} \eps \cdot \area(B_C)$ and $\vol(D) \geq 2^{-O(n)} \eps \cdot \area(B_D)$. Thus,
\begin{equation} \label{eq:mah1}
    \vol(C) \cdot \vol(D) 
        ~ \geq ~ 2^{-O(n)} \eps^2 \cdot \area(B_C) \cdot \area(B_D)
        ~ \geq ~ 2^{-O(n)} \eps^{n+1} \cdot \area(\Upsilon^*) \cdot \area(B_D).
\end{equation}

By Lemma~\ref{lem:sandwich-dualcaps}, $\Upsilon \subseteq B_{\Delta}$, where $B_\Delta = 5 \Delta(B_D)$. Recalling from Lemma~\ref{lem:vol-diffbody} that $\area(\Delta(B_D)) \leq 4^{n-1} \cdot \area(B_D)$, we have
\[
    \area(\Upsilon) 
        ~ \leq ~ \area(B_{\Delta}) 
        ~ =    ~ 5^{n-1} \cdot \area(\Delta(B_D)) 
        ~ \leq ~ 5^{n-1} \cdot 4^{n-1} \cdot \area(B_D) 
        ~ \leq ~ 2^{O(n)} \cdot \area(B_D).
\]
Substituting this bound into Eq.~\eqref{eq:mah1}, we obtain
\[
    \vol(C) \cdot \vol(D) 
        ~ \geq ~ 2^{-O(n)} \eps^{n+1} \cdot \area(\Upsilon^*) \cdot \area(\Upsilon) 
        ~ \geq ~ 2^{-O(n)} \eps^{n+1} \cdot \omega_{n-1}^2,
\]
where we have applied Lemma~\ref{lem:mahler-bounds} to lower bound the Mahler volume in the last step. Since $K$ is well-centered, it follows from Lemma~\ref{lem:centroid} that $K$ satisfies the Santal{\'o} property, that is, $\vol(K) \cdot \vol(K^*) \leq 2^{O(n)} \cdot \omega_n^2$. Recalling the definition of $\omega_n$ from Section~\ref{s:centrality}, we have $\omega_{n-1} / \omega_n = \Theta(\sqrt{n})$. Thus
\[
    \vol_K(C) \cdot \vol_{K^*}(D) 
        ~ \geq ~ 2^{-O(n)} \eps^{n+1},
\]
as desired.
\end{proof}

Finally, we present the main ``take-away'' of this section. This lemma shows that the bound on the product of volumes from the previous lemma holds within the neighborhood of the ray, specifically to any shrunken Macbeath region that intersects the ray.

\begin{lemma} \label{lem:mahler-mac}
Let parameter $\eps$, convex body $K$ and cap $C$ of $K$ be as defined in Lemma~\ref{lem:vol-product}. Suppose that the ray $r$ shot from the origin orthogonal to the base of $C$ intersects a Macbeath region $M^{1/5}(x)$ of $K^*$, where $\ray(x) = \eps$ (see Figure~\ref{f:mahlermac}). Then 
\[
    \vol_K(C) \cdot \vol_{K^*}(M^{1/5}(x)) 
        ~ \geq ~ 2^{-O(n)} \eps^{n+1}.
\]
\end{lemma}

\begin{figure}[htbp]
  \centerline{\includegraphics[scale=.8]{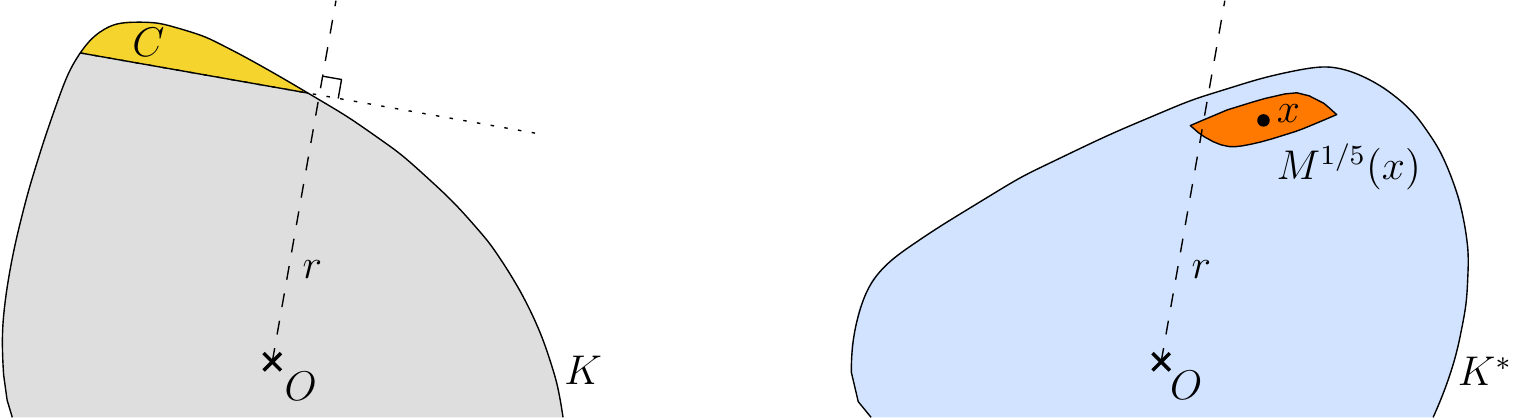}}
  \caption{\label{f:mahlermac}Statement of Lemma~\ref{lem:mahler-mac}.}
\end{figure}

\begin{proof}
Let $y$ be a point in the intersection of the ray $r$ with $M^{1/5}(x)$ and let $D$ denote the minimum volume cap of $K^*$ that contains $y$. Since $M^{1/5}(y) \cap M^{1/5}(x) \neq \emptyset$, by Lemma~\ref{lem:mac-mac}, we have $M^{1/5}(y) \subseteq M^{4/5}(x)$. Thus $\vol(M^{1/5}(x)) \geq 2^{-O(n)} \cdot \vol(M^{1/5}(y))$. Also, by Corollary~\ref{cor:min-vol-cap2}, we have $\vol(M^{1/5}(y)) \geq 2^{-O(n)} \cdot \vol(D)$. Thus $\vol(M^{1/5}(x)) \geq 2^{-O(n)} \cdot \vol(D)$. To complete the proof, it suffices to show the inequality given in the statement of the lemma with $D$ in place of $M^{1/5}(x)$. By Lemma~\ref{lem:core-ray}, we have $\ray(y) \geq \ray(x) / 2$, and by Lemma~\ref{lem:raydist-width}, we have $\width(D) \geq \ray(y)$. Thus $\width(D) \geq \ray(x)/2 = \eps / 2$. Applying Lemma~\ref{lem:vol-product} on caps $C$ and $D$, the desired inequality now follows.
\end{proof}

\section{Covers of Convex Bodies} \label{s:cover}

As mentioned earlier, we employ a Macbeath region-based adaptation of $(c,\eps)$-coverings in our solution to approximate CVP. Since our construction will involve composing coverings of various regions of $K$, we define our coverings in the following restricted manner. Let $K \subseteq \RE^n$ be a convex body, let $\Lambda$ be an arbitrary subset of $\interior(K)$, and let $c \geq 2$ be any constant. Define a \emph{$\Lambda$-limited $c$-covering} to be a collection $\QQ$ of convex bodies that cover $\Lambda$, such that the $c$-factor expansion of each body about its centroid is contained within $K$. 

Our coverings will be based on Macbeath regions. Given $X \subseteq K$, define $\MM_K^{\lambda}(X) = \{ M_K^{\lambda}(x) : x \in X\}$. Define a \emph{$(K, \Lambda, c)$-MNet} to be any maximal set of points $X \subseteq \Lambda$ such that the shrunken Macbeath regions $\MM_K^{1/4c}(X)$ are pairwise disjoint. Through basic properties of Macbeath regions, we can obtain a covering by suitable expansion as shown in the following lemma, which summarizes the properties of MNets.

\begin{lemma} \label{lem:delone}
Given a convex body $K \subseteq \RE^n$, $\Lambda \subset \interior(K)$, and $c \ge 2$, a $(K,\Lambda,c)$-MNet $X$ satisfies the following properties:
\begin{enumerate}
    \item[$(a)$] (Packing) The elements of $\MM_K^{1/4c}(X)$ are pairwise disjoint.
    \item[$(b)$] (Covering) The union of $\MM_K^{1/c}(X)$ covers $\Lambda$.
    \item[$(c)$] (Buffering) The union of $\MM_K(X)$ is contained within $K$.
\end{enumerate}
\end{lemma}

\begin{proof}
Part~(a) is an immediate consequences of the definition. Part~(c) follows by basic properties of Macbeath regions. To prove part (b), let $\lambda = 1/c$ and consider any point $y \in \Lambda$. By maximality, there is $x \in X$ such that $M^{\lambda/4}(x)$ overlaps $M^{\lambda/4}(y)$. By Lemma~\ref{lem:mac-mac}, $M^{\lambda/4}(y) \subseteq M^{\lambda}(x)$, which implies that $y \in M^{\lambda}(x)$.
\end{proof}

Observe that property~(b) implies that if $X$ is a $(K, \Lambda, c)$-MNet, then $\MM_K^{1/c}(X)$ is a $\Lambda$-limited $c$-covering. Further, recalling that $K_{\eps} = (1+\eps)K$, if $X$ is a $(K_{\eps}, K, c)$-MNet, then $\MM_K^{1/c}(X)$ is a $(c,\eps)$-covering of $K$ (see Figure~\ref{f:cover}).

\begin{figure}[htbp]
\centering
\includegraphics[scale=.8,page=1]{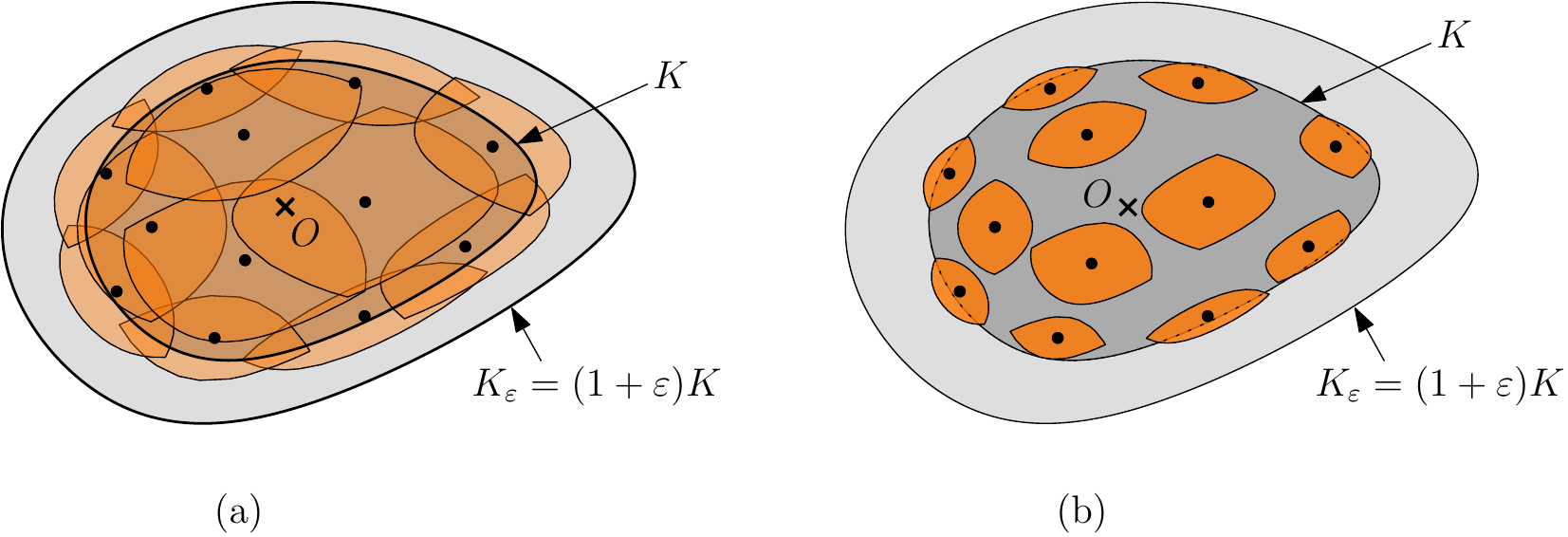}
\caption{\label{f:cover} (a) A $(c,\eps)$-covering of $K$ by Macbeath regions. (b) The corresponding maximal set of disjoint Macbeath regions.}
\end{figure}

\subsection{Instance Optimality} \label{s:instance-opt}

In this section we show that an MNet for $K_{\eps}$ naturally generates an \emph{instance optimal} $(2,\eps)$-covering in the sense that its size cannot exceed that of any $(2,\eps)$-covering of $K$ by a factor of $2^{O(n)}$ (Lemma~\ref{lem:cover-inst} and Theorem~\ref{thm:cover-inst}). It is worth noting that this fact holds irrespective of the location of the origin in $\interior(K)$. In other words, we require no centrality assumptions for this result.

We begin with two lemmas that are straightforward adaptations of lemmas in \cite{NaV22}. The first lemma shows that one incurs a size penalty of only $2^{O(n)}$ by restricting to $c$-coverings to centrally symmetric convex bodies. The second shows that a constant change in the expansion factor results in a similar penalty.

\begin{lemma} \label{lem:cover-sym}
Let $c \ge 2$ be a constant. Let $Q \subseteq \RE^n$ be a convex body with its centroid at the origin. There exists a set of $2^{O(n)}$ centrally symmetric convex bodies which together cover $Q$, such that the central $c$-expansion of any of these bodies is contained within $2 Q$.
\end{lemma}

\begin{proof}
Let $R = M_Q(O) = Q \cap -Q$, and let $R' =  \frac{1}{c} R$ and $R'' = \frac{1}{2c} R$. Clearly, all these bodies are centrally symmetric about the origin. By Lemma~\ref{lem:centroid}, $\vol(R) \geq 2^{-O(n)} \vol(Q)$, and since $c$ is a constant, the volumes of $R'$ and $R''$ are similarly bounded. Let $X \subset Q$ be a maximal discrete set of points such that the translates $X \oplus R'' = \{x + R'' : x \in X\}$ are pairwise disjoint. We will show that the bodies $X \oplus R'$ satisfy the lemma.

To establish the expansion property, observe that for all $x \in X$, $x + c R' = x + R \subseteq Q \oplus R \subseteq 2 Q$. To prove the size bound, by disjointness we have
\[
    |X| \cdot \vol(R'') 
        ~ \leq ~ \vol(2 Q) 
        ~ \leq ~ 2^{O(n)} \vol(Q) 
        ~ \leq ~ 2^{O(n)} \vol(R''),
\]
and therefore $|X| = 2^{O(n)}$. Finally, to prove coverage, consider any $y \in Q$. By maximality there exists $x \in X$ such that $x + R''$ overlaps $y + R''$. Since $c \geq 2$, it follows that $y \in x + 2 R'' = x + R'$.
\end{proof}

\begin{lemma} \label{lem:cover-transform}
Let $K \subseteq \RE^n$ be a convex body, let $\Lambda \subset \interior(K)$, and let $c \ge 2$ be a constant. Let $\QQ$ be a $\Lambda$-limited $c$-covering with respect to $K$. For any constant $c' \ge 2$, there exists a $\Lambda$-limited $c'$-covering with respect to $K$ consisting of centrally symmetric convex bodies whose size is at most $2^{O(n)} |\QQ|$.
\end{lemma}

\begin{proof}
By Lemma~\ref{lem:cover-sym}, we can replace each body $Q \in \QQ$ by a set of $2^{O(n)}$ centrally symmetric convex bodies which together cover $Q$ and such that the $c'$-expansion of any of these bodies is contained within the 2-expansion of $Q$ (about its centroid). It is easy to see that the resulting set of bodies is a $\Lambda$-limited $c'$-cover with respect to $K$ with the desired size.
\end{proof}

We are now ready to show that a $(K,\Lambda,c)$-MNet can be used to generate an instance-optimal limited covering.

\begin{lemma} \label{lem:cover-inst}
Let $K \subseteq \RE^n$ be a convex body, let $\Lambda \subset \interior(K)$, and let $c \geq 2$ be a constant. Let $X$ be a $(K,\Lambda,c)$-MNet, and let $\MM = \MM_K^{1/c}(X)$ be the associated $\Lambda$-limited $c$-covering with respect to $K$. Given any $\Lambda$-limited $c$-covering $\QQ$ with respect to $K$, $|\MM| \leq 2^{O(n)} |\QQ|$.
\end{lemma}

\begin{proof}
By Lemma~\ref{lem:cover-transform}, there exists a $\Lambda$-limited 5-covering with respect to $K$ consisting of at most $2^{O(n)} |\QQ|$ centrally symmetric convex bodies. Let $\QQ'$ denote this covering, and let $Y$ denote the set of centers of these bodies. Consider any $Q \in \QQ'$, and let $y$ denote its center. By definition, $M(y) = M_K(y)$ is the largest centrally symmetric body centered at $y$ that is contained within $K$. Since $Q$ is a centrally symmetric convex body whose 5-expansion about $y$ is contained within $K$, it follows that $Q \subseteq M^{1/5}(y)$. Therefore, $\MM^{1/5}(Y)$ is a $\Lambda$-limited 5-covering of the same cardinality as $\QQ'$. 

By the packing property of Lemma~\ref{lem:delone}, the Macbeath regions $\MM^{1/4c}(X)$ are pairwise disjoint. To relate these two coverings, assign each $x \in X$ to any $y \in Y$ such that $x \in M^{1/5}(y)$. We will show that at most $2^{O(n)}$ elements of $X$ are assigned to any $y \in Y$. Assuming this for now, we have 
\[
    |\MM|
        ~ =    ~ |X| 
        ~ \leq ~ 2^{O(n)} |Y| 
        ~ =    ~ 2^{O(n)} |\QQ'| 
        ~ \leq ~ 2^{O(n)} |\QQ|,
\]
thus completing the proof.

To prove the assertion, consider any $x \in X$ assigned to some $y \in Y$. Since $M^{1/5}(x) \cap M^{1/5}(y) \neq \emptyset$, by Lemma~\ref{lem:mac-mac} and the fact that $c \geq 2$, we have
\[
    M^{1/4c}(x)
        ~ \subseteq ~ M^{1/5}(x) 
        ~ \subseteq ~ M^{4/5}(y).
\]
Lemma~\ref{lem:mac-mac} also implies that $M^{1/5}(y) \subseteq M^{4/5}(x)$, and so $\vol(M^{1/4c}(x)) \geq 2^{-O(n)} \vol(M^{4/5}(y))$. Since the Macbeath regions of $M^{1/4c}(X)$ are pairwise disjoint, by a simple packing argument, the number of points of $X$ assigned to any $y \in Y$ is at most $2^{O(n)}$, as desired.
\end{proof}

Recall that a $K$-limited $c$-covering with respect to $K_{\eps} = (1+\eps) K$ is a $(c,\eps)$-covering for $K$. Applying the above lemma in this case, we obtain the main result of this section.

\begin{theorem} \label{thm:cover-inst}
Let $0 < \eps \leq 1$, let $K \subseteq \RE^n$ be a convex body such that $O \in \interior(K)$, and let $c \ge 2$ be a constant. Let $X$ be a $(K_{\eps}, K, c)$-MNet, and let $\MM = \MM_{K_\eps}^{1/c}(X)$ be the associated $(c,\eps)$-covering with respect to $K$. Given any $(c,\eps)$-covering $\QQ$ with respect to $K$, $|\MM| \leq 2^{O(n)} |\QQ|$.
\end{theorem}

\subsection{Worst-Case Optimality} \label{s:worst-opt}

Our main result in this section, given in Lemma~\ref{lem:cover-worst}, establishes the existence of a $(c,\eps)$-covering of size $2^{O(n)}/\eps^{(n-1)/2}$. This directly implies Theorem~\ref{thm:cover-worst}. Before presenting this result, it will be useful to first establish a bound on the maximum number of disjoint Macbeath regions associated with $\Theta(\eps)$-width caps. The proof is based on the relationship between caps in $K$ and $K^*$.

Let $K \subseteq \RE^n$ be a well-centered convex body. Given $0 < \eps \leq \frac{1}{32}$, let $\Lambda \subseteq K$ denote the centroids of the bases of all caps whose relative widths are between $\eps$ and $2\eps$. Given a constant $c \geq 2$, let $X$ be a $(K, \Lambda, c)$-MNet, and let $\MM(X) = \MM_K^{1/c}(X)$ be the associated covering. We will show that $|X| \le 2^{O(n)} / \eps^{(n-1)/2}$, which will imply a similar bound on the size of the associated $\Lambda$-limited $c$-covering. 

Recall that for any region $\Lambda \subseteq K$, its relative volume is $\vol_K(\Lambda) = \vol(\Lambda)/\vol(K)$. Let $t = \eps^{(n+1)/2}$. Define $X_{\geq t} = \{x \in X : \vol_K(M_K^{1/c}(x)) \geq t\}$ to be the centers of the ``large'' Macbeath regions in the covering of relative volume at least $t$, and let $X_{< t} = X \setminus X_{\geq t}$ denote the centers of the remaining ``small'' Macbeath regions.

To bound the number of small Macbeath regions, we will make use of the polar body $K^*$. Let $\Lambda'$ denote the boundary of $(1-\eps) K^*$. Let $Y$ be a $(K^*, \Lambda', 5)$-MNet, and let $\MM(Y) = \MM_{K^*}^{1/5}(Y)$ be the associated covering. Let $t' = 2^{-O(n)} \eps^{(n+1)/2}$, where the constant hidden in $O(n)$ is sufficiently large, and analogously define $Y_{\geq t'} = \{y \in Y : \vol_{K^*}(M_{K^*}^{1/5}(y)) \geq t'\}$ to be the set of centers of the ``large'' Macbeath regions in the polar covering $\MM(Y)$ whose relative volume is at least $t'$.

The following lemma summarizes the essential properties of the resulting Macbeath regions. 

\begin{lemma} \label{lem:layer}
Given a well-centered convex body $K \subseteq \RE^n$, $0 < \eps \leq \frac{1}{32}$, constant $c \ge 2$, and the entities $\Lambda$, $\Lambda'$, $X$, $Y$, $t$, and $t'$ defined above, the following hold:
\begin{enumerate}
    \item[$(a)$] The regions $\MM_K^{1/c}(X)$ are contained in $\Lambda_K(\eps) = K \setminus (1-4\eps)K$, and $\vol_K(\Lambda_K(\eps)) = O(n \eps)$.
    \item[$(b)$] For any $x \in X_{\geq t}$, $\vol_K(M^{1/c}(x)) \geq \eps^{(n+1)/2}$, and $|X_{\geq t}| \le 2^{O(n)} / \eps^{(n-1)/2}$.
    \item[$(c)$] The regions $\MM_{K^*}^{1/5}(Y)$ are contained in $\Lambda_{K^*}(\eps) = K^* \setminus (1-2\eps)K^*$, and $\vol_{K^*}(\Lambda_{K*}(\eps)) = O(n \eps)$.
    \item[$(d)$] For any $y \in Y_{\geq t'}$, $\vol_{K^*}(M^{1/5}(y)) \geq 2^{-O(n)} \eps^{(n+1)/2}$, and $|Y_{\geq t'}| \le 2^{O(n)} / \eps^{(n-1)/2}$.
    \item[$(e)$] For any $x \in X_{< t}$, there is $y \in Y_{\geq t'}$ such that for any point $z \in M^{1/5}(y)$, we have $M^{1/c}(x) \subseteq C_z^{32}$, and $\vol(M^{1/c}(x)) \geq 2^{-O(n)} \vol(C_z^{32})$, where $C_z \subseteq K$ is $z$'s $\eps$-representative cap.
    \item[$(f)$] $|X| \le 2^{O(n)} / \eps^{(n-1)/2}$.
\end{enumerate}
\end{lemma}

\begin{proof}
To prove~(a), let $x$ be any point of $X$ and let $M_x = M^{1/c}(x)$ be the associated covering Macbeath region. Because $X$ is a $(K, \Lambda, c)$-MNet, $M_x$ is centered at the centroid of the base of a cap $C_x$ of width between $\eps$ and $2\eps$. Since $c \geq 1$, by Lemma~\ref{lem:mac-cap-var}, $M_x \subseteq C_x^2$. As $C_x^2$ has width at most $4\eps$, it follows that $C_x^2 \subseteq \Lambda_K(\eps)$, and so too is $M_x$. Clearly, $\vol_K(\Lambda_K(\eps)) = 1 - (1-4\eps)^n = O(n\eps)$. 

To prove~(b), observe that the Macbeath regions $\MM^{1/4c}(X_{\geq t})$ are pairwise disjoint, and each has relative volume at least $t/4^n \geq 2^{-O(n)} \eps^{(n+1)/2}$. By a simple packing argument, $|X_{\geq t}| \leq \vol_K(\Lambda_K(\eps)) / (t/4^n) \leq 2^{O(n)} / \eps^{(n-1)/2}$.

To prove~(c), let $y$ be any point of $Y$ and let $M_y = M^{1/5}(y)$ be the associated covering Macbeath region. Since $y$ lies on the boundary of $(1-\eps) K^*$, $y$ lies on the base of a cap $C_y$ of $K^*$ induced by the supporting hyperplane of $(1-\eps) K^*$. By Lemma~\ref{lem:mac-cap-var}, $M_y \subseteq C_y^2$. Since $C_y^2$ has width $2\eps$, it follows that $C_y^2 \subseteq \Lambda_{K*}(\eps)$, and so too is $M_y$. Also, $\vol_{K^*}(\Lambda_{K*}(\eps)) = 1 - (1-2\eps)^n = O(n\eps)$.

To prove~(d), observe that by Lemma~\ref{lem:delone}, the Macbeath regions $\MM^{1/(4 \cdot 5)}(Y_{\geq t'})$ are pairwise disjoint, and each has relative volume at least $t'/4^n = 2^{-O(n)} \eps^{(n+1)/2}$. By a simple packing argument, $|Y_{\geq t'}| \leq \vol_{K^*}(\Lambda_{K*}(\eps)) / (t'/4^n) \leq 2^{O(n)} / \eps^{(n-1)/2}$. 

To prove~(e), let $x$ be any point of $X_{< t}$ and let $M_x = M^{1/c}(x)$ be the associated covering Macbeath region. As in~(a), $M_x$ is centered at the centroid of the base of a cap $C_x$ of width between $\eps$ and $2\eps$. Since $c$ is a constant, by Lemma~\ref{lem:min-vol-cap2}, $\vol(C_x) \leq 2^{O(n)} \vol(M_x)$. Since $\vol_K(M_x) \leq t = \eps^{(n+1)/2}$, we have $\vol_K(C_x) \leq 2^{O(n)} \eps^{(n+1)/2}$. 

In the polar, consider the ray $r$ shot from the origin orthogonal to the base of $C_x$. This ray will intersect some covering Macbeath region $M_y = M^{1/5}(y)$, for some $y \in Y$. We will show that $y$ satisfies all the properties given in part~(e). As $K$ is well-centered, we can apply the Mahler-like volume relation from Lemma~\ref{lem:mahler-mac} to obtain $\vol_K(C_x) \cdot \vol_{K^*}(M_y) \geq 2^{-O(n)} \eps^{n+1}$. Using the upper bound on $\vol_K(C_x)$ shown above, it follows that $\vol_{K^*}(M_y) \geq 2^{-O(n)} \eps^{(n+1)/2}$. Thus, $y \in Y_{\geq t'}$.

It is easy to verify that the preconditions of Corollary~\ref{cor:sandwich} are satisfied where $C_x$ plays the role of $C$, $M_y$ plays the role of $M^{1/5}(y)$, and $z$ is any point in $M_y$. It follows that the caps $C_x$ and $C_z$ are 16-similar, that is, $C_x \subseteq C_z^{16}$ and $C_z \subseteq C_x^{16}$. By Lemma~\ref{lem:mac-cap-var}, $M_x \subseteq C_x^2$, and by Lemma~\ref{lem:cap-containment-exp}, $C_x^2 \subseteq C_z^{32}$. Thus $M_x \subseteq C_z^{32}$. Also, since $C_z \subseteq C_x^{16}$, it follows from Lemma~\ref{lem:cap-exp} that $\vol(C_x) \geq 2^{-O(n)} \vol(C_z)$. By Lemma~\ref{lem:min-vol-cap2}, $\vol(M_x) \geq 2^{-O(n)} \vol(C_x)$. Thus $\vol(M_x) \geq 2^{-O(n)} \vol(C_z) \ge 2^{-O(n)} \vol(C_z^{32})$, which establishes~(e).

Finally, to prove~(f), observe that in light of~(b), it suffices to show that $|X_{< t}| \le 2^{O(n)} / \eps^{(n-1)/2}$. This quantity can be bounded by the following charging argument. For each $y \in Y_{\geq t'}$, we say that it \emph{charges} all the points $x \in X$ whose Macbeath region $M^{1/4c}(x)$ is contained in $C_y^{32}$ and whose volume is at least $2^{-O(n)} \vol(C_y^{32})$, where the constant hidden in $O(n)$ is sufficiently large. Note that any point of $Y_{\geq t'}$ charges at most $2^{O(n)}$ points of $X$. Applying part (e), it follows that every $x \in X_{< t}$ is charged by some $y \in Y_{\geq t'}$. Since $|Y_{\geq t'}| \leq 2^{O(n)} / \eps^{(n-1)/2}$ and each point of $Y_{\geq t'}$ charges at most $2^{O(n)}$ points of $X$, it follows that $|X_{< t}| \leq 2^{O(n)} / \eps^{(n-1)/2}$, which completes the proof.
\end{proof}

We are now ready to present the main result of this section. Recall that $K \subseteq \RE^n$ is a well-centered convex body. Given $0 < \eps \leq 1$, define a \emph{layered decomposition} of $K$ as follows. Recalling that $K_{\eps} = (1+\eps)K$, for each $x \in K$, define its \emph{width}, denoted $\width(x)$, to be the width of the associated minimum volume cap of $K_{\eps}$. Since $\ray_{K_{\eps}}(x) \geq \eps / (1 + \eps) \geq \eps / 2$, it follows from Lemma~\ref{lem:raydist-width} that $\width(x) \geq \eps / 2$. Let $\beta$ be a sufficiently small constant, and let $k_0 = \ceil{\log\frac{\beta}{\eps}}$. For $0 \leq i \leq k_0$, define the layer $i$ be the set of points $x \in K$ such that $\width(x) \in [2^{i-1},2^i)\eps$. Define layer $k_0 + 1$ to be the set of remaining points of $K$, which have width at least $\beta$. Note that the number of layers is $O(\log\frac{1}{\eps})$. 

\begin{lemma} \label{lem:cover-worst}
Let $0 < \eps \leq 1$, let $K \subseteq \RE^n$ be a well-centered convex body, and let $c \ge 2$ be a constant. Let $X$ be a $(K_{\eps},K,c)$-MNet, and let $\MM = \MM_{K_{\eps}}^{1/c}(X)$. Then $\MM$ is a $(c,\eps)$-covering for $K$ consisting of at most $2^{O(n)} / \eps^{(n-1)/2}$ centrally symmetric convex bodies.
\end{lemma}

\begin{proof} 
By Lemma~\ref{lem:delone}, $\MM$ is a $(c,\eps)$-covering for $K$. We will bound the size of the covering by partitioning the points of $X$ based on the layered decomposition (defined above) and then use Lemma~\ref{lem:layer} to bound the number of points in each layer. 

For $0 \leq i \leq k_0$, let $X_i$ be subset of points of $X$ that are in layer $i$. Since $K$ is well-centered, $K_{\eps}$ is also well-centered. By Lemma~\ref{lem:layer}(f), $|X_i| \leq 2^{O(n)} /(2^i \eps)^{(n-1)/2}$. Summing $|X_i|$ over all layers $0$ to $k_0$ we have at most $2^{O(n)} / \eps^{(n-1)/2}$ points in all these layers.

It remains only to bound $|X_{k_0+1}|$. Consider the set $\MM_{K_{\eps}}^{1/4c}(X_{k_0+1})$ of the associated packing Macbeath regions. By Lemma~\ref{lem:delone}, these Macbeath regions are pairwise disjoint. Recall that the minimum volume cap of any point in $X_{k_0+1}$ has width at least $\beta$ (used in the definition of $k_0$). Hence by Lemma~\ref{lem:wide-cap} (and the fact that $c$ is a constant), each of these Macbeath regions has relative volume of at least $2^{-O(n)}$. By a simple packing argument, it follows that $|X_{k_0+1}| \leq 2^{O(n)}$, which completes the proof. 
\end{proof}

\section{Applications: Banach-Mazur Approximation} \label{s:approx-BM}

In this section we show that the convex hull of the centers of any $(c,\eps)$-covering implies the existence of an approximating polytope in the Banach-Mazur distance. The main result is given in the following lemma. Combining this with our covering from Theorem~\ref{thm:cover-worst} establishes Theorem~\ref{thm:approx-BM}.

\begin{lemma} \label{lem:approx-BM}
Let $0 < \eps < 1$, let $K \subseteq \RE^n$ be a well-centered convex body, and let $c \ge 2$ be a constant. Let $X$ be the set of centers of any $(c,\eps')$-covering of $K(1 + \eps/c)$, where $\eps' = \frac{1+\eps}{1+\eps/c} - 1$. Then $K \subset \conv(X) \subset K(1+\eps)$.
\end{lemma}

\begin{proof}
Let $\MM$ denote the covering mentioned in the statement of the lemma. By definition, the bodies of $\MM$ together cover $K(1+\eps/c)$ and the $c$-expansion of any such body about its center is contained within $K (1+\eps)$. Since each body of $\MM$ is contained within $K(1+\eps)$, it follows that $X \subset K(1+\eps)$ and so $\conv(X) \subset K(1+\eps)$. To prove that $K \subset \conv(X)$, it suffices to show that there is a point of $X$ in every cap of $K(1+\eps)$ defined by a supporting hyperplane of $K$. 

\begin{figure}[htb]
  \centerline{\includegraphics[scale=.8]{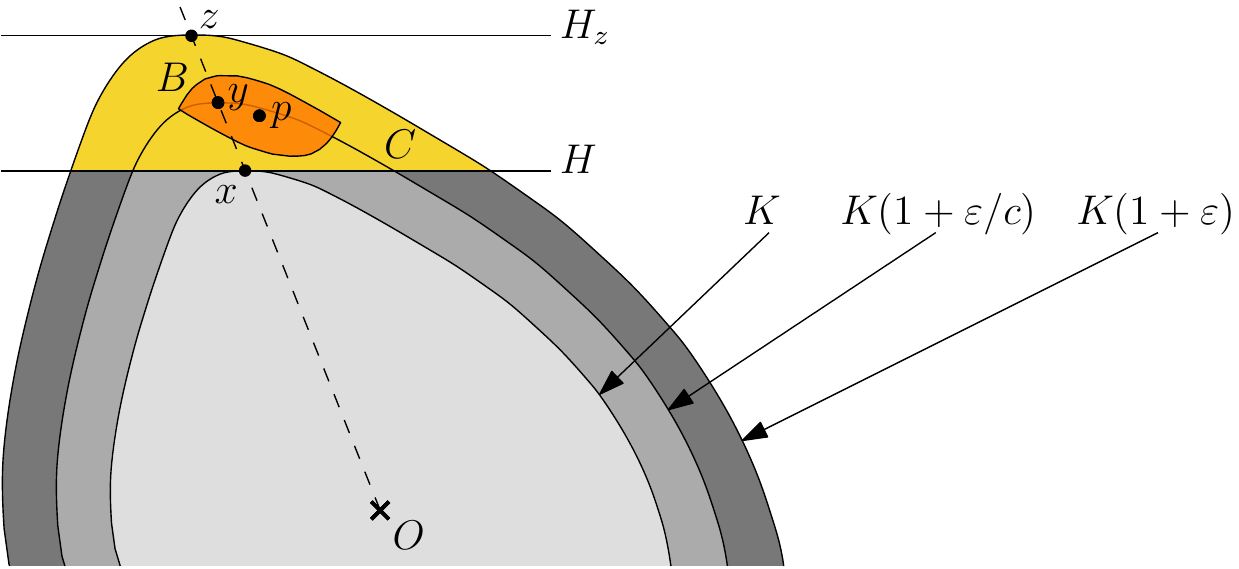}}
  \caption{\label{f:thm2}Proof of Lemma~\ref{lem:approx-BM}.}
\end{figure}

Let $C$ be a cap of $K(1+\eps)$ defined by a supporting hyperplane $H$ of $K$. Let $x$ be a point at which $H$ touches $K$. For the sake of concreteness, assume that $H$ is horizontal and $K$ lies below $H$. Consider the ray emanating from the origin passing through $x$. Suppose that this ray intersects the boundary of $K(1+\eps/c)$ at $y$ and the boundary of $K(1+\eps)$ at $z$. Let $H_z$ denote the supporting hyperplane of $K(1+\eps)$ at $z$. Clearly $H_z$ is parallel to $H$ and the distance between $H$ and $H_z$ is $c$ times the distance between $y$ and $H$. 

Consider any body $B$ of $\MM$ that contains point $y$. We claim that the center $p$ of the body $B$ is contained within $C$. By our earlier remarks, $p \in K(1+\eps)$. Thus, we only need to show that $p$ cannot lie below $H$. To see this, recall that the body formed by expanding $B$ about its center $p$ by a factor of $c$ is contained within $K(1+\eps)$. In particular, the point $p + c (y-p) \in K(1+\eps)$. However, if $p$ lies below $H$, then the point $p + c (y-p)$ would lie above $H_z$, and hence outside $K(1+\eps)$. It follows that $p$ cannot lie below $H$, which completes the proof.
\end{proof}

By Lemma~\ref{lem:cover-worst}, there exists a $(c,\eps')$-covering $\MM$ for $K(1 + \eps/c)$ consisting of at most $2^{O(n)} / (\eps')^{(n-1)/2}$ centrally symmetric convex bodies. The bound on vertices in Theorem~\ref{thm:approx-BM} now follows immediately from the above lemma (setting $P = \conv(X)$ and noting that $\eps' = \Theta(\eps)$), and the bound on facets follows via polarity and scaling by a factor of $(1+\eps)$. 

\section{Applications: Approximate CVP and IP} \label{s:apps}

\subsection{Preliminaries} \label{s:apps-prelim}

An $n$-dimensional lattice $L \subseteq \RE^n$ is the set of all integer linear combinations of a basis $b_1, \ldots, b_n$ of $\RE^n$. Given a lattice $L$, a convex body $K$ and a target $t \in \RE^n$, the \emph{closest vector problem} (CVP) seeks to find a closest vector in $L$ to $t$ under $\|\cdot\|_K$.  Given a parameter $\eps > 0$, the \emph{$(1+\eps)$-approximate CVP problem} seeks to find any lattice vector whose distance to $t$ under $\|\cdot\|_K$ is at most $(1+\eps)$ times the true closest.

We employ a standard computational model in our $(1+\eps)$-CVP algorithm. Given reals $0 < r \leq r'$ and $x \in \RE^n$, we say that a convex body $K \subseteq \RE^n$ is \emph{$(x,r,r')$-centered} if $x + r B_2^n \subseteq K \subseteq x + r' B_2^n$, where $B_2^n$ is the unit Euclidean ball centered at the origin. We assume that the convex body $K$ inducing the norm is $(O,r,r')$-centered, where both $r$ and $r'$ are given explicitly as inputs.
We assume that the basis vectors of the lattice $L$ are presented as an $n \times n$ matrix over the rationals.  Input size is measured as the total number of bits used to encode $r$, $r'$, $t$, and $\eps$ and the basis vectors of $L$ (all rationals). 

Following standard conventions, we assume that access to $K$ is provided through a \emph{membership oracle}, which on input $x \in \RE^n$ returns 1 if $x \in K$ and 0 otherwise. Our algorithms apply more generally where $K$ is presented using a \emph{weak membership oracle}, which takes an extra parameter $\delta > 0$ and only needs to return the correct answer when $x$ is at Euclidean distance at least $\delta$ from the boundary of $K$. 

In the oracle model of computation, the running time is measured by the number of oracle calls and bit complexity of arithmetic operations. Note that the running time of our $(1+\eps)$-CVP algorithm will be exponential in the dimension $n$. We will follow standard practice and suppress polynomial factors in $n$ and the input size. We will also simplify the presentation by expressing our algorithms assuming exact oracles, but the adaptation to weak oracles is straightforward.

Our approach to approximate CVP follows one introduced by Eisenbrand {\etal}~\cite{EHN11} for $\ell_{\infty}$ and later extended in a number of works~\cite{NaV22,EiV22,RoV22}, which employs coverings of $K$. Given any constant $c \geq 2$, a \emph{$(c,\eps)$-covering} of an $(O,r,r')$-centered convex body $K$ is a collection $\QQ$ of convex bodies, such that a factor-$c$ expansion of each $Q \in \QQ$ about its centroid lies within $K_{\eps}$. Nasz{\'o}di and Venzin showed that a $(2,\eps)$-covering of $K$ can be used to boost the approximation factor of any $2$-CVP solver for general norms.

\begin{lemma}[Nasz{\'o}di and Venzin~\cite{NaV22}] \label{lem:boost}
Let $L$ be a lattice and let $K$ be an $(O,r,r')$-centered convex body. Given a $(2,\eps)$-covering of $K$ consisting of $N$ centrally symmetric convex bodies, we can solve $(1 + 7\eps)$-CVP under $\|\cdot\|_K$ with $\widetilde{O}(N)$ calls to a 2-CVP solver for norms (where $\widetilde{O}$ conceals polylogarithmic factors). 
\end{lemma}

\subsection{CVP Algorithm} \label{s:apps-cvp}

As in Lemma~\ref{lem:cover-worst}, let $K \subseteq \RE^n$ be a well-centered convex body. In this section, we present our algorithm for computing a $(1+\eps)$-approximation to the closest vector (CVP) under the norm defined by $K$. 

Given a convex body $K \subseteq \RE^n$, $0 < \eps \leq 1$, and a constant $c \geq 2$, a \emph{$(c,\eps)$-enumerator} is a procedure that outputs the elements of a $(c,\eps)$-covering for $K$. Each of the elements of the covering is represented as an oracle for an $(a,r,r')$-centered convex body, where $a$, $r$, and $r'$ are given explicitly in the output (as rationals). Our enumerator will be randomized in the Monte Carlo sense, meaning that it achieves a stated running time, but the output may fail to be a $(c,\eps)$-covering with some given probability. Define an enumerator's \emph{overhead} to be its total running time divided by the number of elements output, and its \emph{space complexity} to be the amount of memory it needs.

Our enumerator is based on constructing hitting sets for coverings associated with certain MNets. The following lemma will be useful. 

\begin{lemma} \label{lem:hitting-set}
Let $K \subseteq \RE^n$ be a convex body, $\Lambda \subset \interior(K)$, and $c \geq 2$. Let $X$ be a $(K,\Lambda,4c)$-MNet and let $\MM = \MM_K^{1/4c}(X)$ be the associated covering. Let $Y$ be any hitting set for $\MM$ in the sense that for each $M \in \MM$, $Y \cap M \neq \emptyset$. Then $\MM_K^{1/c}(Y)$ is a $\Lambda$-limited $c$-covering with respect to $K$.
\end{lemma}

\begin{proof} 
Since $c > 1$, the $c$-expansion of any Macbeath region of $M^{1/c}(Y)$ is contained within $K$. To prove the covering property, let $z$ be any point of $\Lambda$. By Lemma~\ref{lem:delone}, there is a point $x \in X$ such that $z \in M^{1/4c}(x)$. Let $y$ be a point of $Y$ that is contained in $M^{1/4c}(x)$. Since $M^{1/4c}(x) \cap M^{1/4c}(y) \neq \emptyset$, by Lemma~\ref{lem:mac-mac}, $M^{1/4c}(x) \subseteq M^{1/c}(y)$. Thus $z \in M^{1/c}(y)$. It follows that $M^{1/c}(Y)$ is a $\Lambda$-limited $c$-covering with respect to $K$.
\end{proof}

The following lemma shows that membership oracles for $K$ can be extended to its polar as well as Macbeath regions and caps that are $\eps$-deep.

\begin{lemma} \label{lem:oracle}
Given an $(O,r,r')$-centered convex body $K$, specified by a weak membership oracle, in time polynomial in $n$, $\log\frac{1}{\eps}$, and $\log\frac{r'}{r}$  we can do the following:
\begin{enumerate}
    \item[$(i)$] Construct a weak membership oracle for $K^*$.
    \item[$(ii)$] Given a point $x \in K$ such that $\ray(x) \geq \eps$, construct a weak membership oracle for $M^{\lambda}_K(x)$ for any constant $\lambda > 0$.
    \item[$(iii)$] Given a hyperplane $h$ intersecting $K$ which induces a cap $C$ of width at least $\eps$, construct a weak membership oracle for $C$.
\end{enumerate}
\end{lemma}

\begin{proof}
Assertion~(i) follows directly from standard reductions (see Theorem~{4.3.2} and Lemma~{4.4.1} from Gr{\"o}tschel, Lov{\'a}sz, and Schrijver~\cite{GLS88}).
Note that $K^*$ is $\big(O,\frac{1}{r'},\frac{1}{r}\big)$-centered. To prove (ii), note that we can construct a membership oracle for $M(x)$ by using the fact that a point $y \in M(x)$ if and only if $y \in K$ and $2 x - y \in K$. If $\ray(x) \ge \eps$, it is straightforward to show that $M(x)$ is $(x, \Omega(\eps r), r')$-centered. The generalization of this construction to $M^{\lambda}_K(x)$ for any constant $\lambda > 0$ is immediate. Finally, to prove (iii), observe that the membership oracle is easy, but centering is the issue. We first determine the apex $a$ of $C$ (approximately) by finding the supporting hyperplane of $K$ that is parallel to $h$. We let $b$ denote the point midway on the segment $O a$ between base of the cap and $a$. It is easy to show that a Euclidean ball of radius $\Omega(\eps r)$ can be centered at $b$, which is contained within $C$. Thus $C$ is $(b,\Omega(\eps r), 2 r')$-centered.
\end{proof}

We will make use of standard sampling results (see, e.g., \cite{DFK91,Vem05}), which state that given $\eta > 0$, there exists an algorithm that outputs an $\eta$-uniform $X \in K$ using at most $\poly\big(n, \ln \frac{1}{\eta}, \ln \frac{r'}{r}\big)$ calls to a membership oracle for $K$ and arithmetic operations. (A random point $X \in K$ is \emph{$\eta$-uniform} if the total variation distance between the sample $X$ and uniform vector in $K$ is at most $\eta$.) As with membership oracles, it will simplify the presentation to state our constructions in terms of a true uniform sampler, but the generalization is straightforward.

\begin{lemma} \label{lem:sample-cover} 
Given $0 < \eps \leq 1$, constant $c \geq 2$, and an oracle for a convex body $K \subseteq \RE^n$ which is both well-centered and $(O,r,r')$-centered, there exists a randomized $(c,\eps)$-enumerator for $K$, which generates a covering of size 
\[
    2^{O(n)} \cdot \frac{1}{\eps^{(n-1)/2}} \cdot \log\frac{1}{\eps},
\]
such that the cover elements are $(a, O(\eps r), r')$-centered. The enumerator succeeds with probability $1 - 2^{-O(n)}$, and its overhead and space complexity are both polynomial in $n$, $\log\frac{r'}{r}$  and $\log\frac{1}{\eps}$.
\end{lemma}

In our construction, the elements of the covering will be centrally symmetric, and more specifically, the covering element centered at a point $a \in K$ will be a Macbeath region of the form $M_{K_{\eps}}^{1/c'}(a)$, where $c' = O(c)$.

\begin{proof}
Recall the layered decomposition of $K$ described just before Lemma~\ref{lem:cover-worst}. For $0 \le i \le k_0$, layer $i$ consists of points $x \in K$ such that $\width(x) \in [2^{i-1},2^i)\eps$, and layer $k_0+1$ consists of the remaining points $x \in K$. Note that for points in layer $k_0+1$, $\width(x) \ge \beta$. Here $\beta$ is a constant and the number of layers $k_0+2 = O(\log\frac{1}{\eps})$. Let $\Lambda_i$ denote the points in layer $i$. Our enumerator runs in phases, where the $i$-th phase generates elements of a $\Lambda_i$-limited $c$-covering with respect to $K_{\eps}$. Clearly, the elements generated in all the phases together constitute a $(c,\eps)$-covering for $K$. 

For $0 \le i \le k_0$, to describe phase $i$ of the enumerator, it will simplify notation to write $K, \Lambda, \eps$, and $c$ for $K_{\eps}, \Lambda_i, 2^{i-1}\eps$, and $4c$, respectively. Our (new) objective is to generate a $\Lambda$-limited $(c/4)$-covering in this phase. Let $X$ be a $(K,\Lambda,c)$-MNet, let $\MM = \MM_{K}^{1/c}(X)$ be the associated covering, and let $X'$ be a hitting set for $\MM$. By Lemma~\ref{lem:hitting-set}, $\MM_K^{4/c}(X')$ is a $\Lambda$-limited $(c/4)$-covering.

We show how to generate the hitting set $X'$ for $\MM$ along with the elements of $\MM_K^{4/c}(X')$ in the desired form. In addition to the quantities $K, \Lambda, \eps, c, X$ defined above, define also the quantities $\Lambda',Y,t,t'$, as in Lemma~\ref{lem:layer}. By Lemma~\ref{lem:layer}(a), the regions of $\MM$ are contained in $\Lambda_K(\eps) = K \setminus (1-4\eps) K$. Recall the distinction between ``large'' and ``small'' Macbeath regions of $\MM$, based on whether its relative volume is at least $t$. We will use a different strategy for hitting these two kinds of regions. 

First, let us consider the large Macbeath regions. We claim that it suffices to choose $(2^{O(n)} / \eps^{(n-1)/2}) \cdot \log\frac{1}{\eps}$ points uniformly in $\Lambda_K(\eps)$ to hit all the large Macbeath regions with high probability. Before proving this, note that we can sample $\Lambda_K(\eps)$ uniformly by first choosing a point $p$ from the uniform distribution in $K$ and then choosing a point uniformly from the portion of the ray $Op \cap \Lambda_K(\eps)$. Using binary search, we can find such a point with constant probability in $O(\log\frac{r'}{r} + \log\frac{1}{\eps})$ steps. We omit the straightforward details.

To prove the claim, let $M$ be a large Macbeath region. By Lemma~\ref{lem:layer}(a) and (b), $M \subseteq \Lambda_K(\eps)$, $\vol_K(M) \ge \eps^{(n+1)/2}$, and $\vol_K(\Lambda_K(\eps)) = O(n\eps)$. Thus $\vol(M) / \vol(\Lambda_K(\eps)) \geq 2^{-O(n)} \eps^{(n-1)/2}$. Also, by Lemma~\ref{lem:layer}(b), the number of large Macbeath regions is at most  $2^{O(n)}/\eps^{(n-1)/2}$. A standard calculation implies that the probability of failing to hit some large Macbeath region in a layer is no more than $\eps^{O(n)}$.

Next we show how to generate a hitting set for the small Macbeath regions. Intuitively, as these are small, they cannot be stabbed efficiently by uniform sampling in $\Lambda_K(\eps)$. Instead, we will hit them by exploiting the relationship between the small Macbeath regions of $\MM$ and the large Macbeath regions of $\MM' = \MM_{K^*}^{1/5}(Y)$. Recall that $Y$ is a $(K^*,\Lambda',5)$-MNet, where $\Lambda'$ is the boundary of $(1-\eps)K^*$, and the large Macbeath regions of $\MM'$ have volume at least $t' = 2^{-O(n)} \eps^{(n+1)/2}$. Our high-level idea for hitting the small Macbeath regions of $\MM$ is to hit the large Macbeath regions of $\MM'$ and then uniformly sample the associated $\eps$-representative cap of $K$. 

More precisely, we perform $(2^{O(n)} / \eps^{(n-1)/2}) \cdot \log(1/\eps)$ iterations of the following procedure. First, we choose a point $p$ uniformly in $\Lambda_{K^*}(\eps) = K^* \setminus (1-2\eps) K^*$. (This can be done in a manner analogous to uniformly sampling $\Lambda_K(\eps)$, which we described above.) Next, we sample uniformly in the cap $C_p^{32}$, where $C_p$ is $p$'s $\eps$-representative cap in $K$. We claim that this procedure stabs all the small Macbeath regions of $\MM$ with high probability.

To see why, recall from Lemma~\ref{lem:layer}(e) that for any small Macbeath region $M \in \MM$, there is a large Macbeath region $M' \in \MM'$ with the following properties. Let $y$ be any point in $M'$ and let $C_y$ be $y$'s $\eps$-representative cap in $K$. Then $M \subseteq C_y^{32}$ and $\vol(M) \ge 2^{-O(n)} \vol(C_y^{32})$. Also, by properties (c) and (d) of Lemma~\ref{lem:layer}, we have $M' \subseteq \Lambda_{K^*}(\eps)$, $\vol_{K^*}(M') \ge 2^{-O(n)} \eps^{(n+1)/2}$, and  $\vol_{K^*}(\Lambda_{K^*}(\eps)) = O(n\eps)$. It follows that the probability of hitting a fixed small Macbeath region $M$ of $\MM$ in any one trial (\textit{i.e.}, sampling $p$ uniformly in $\Lambda_{K^*}(\eps)$, followed by sampling a point uniformly in the cap $C_p^{32}$) is at least $2^{-O(n)} \eps^{(n-1)/2}$. Also, by Lemma~\ref{lem:layer}(f), the number of small Macbeath regions of $\MM$ is at most  $2^{O(n)}/\eps^{(n-1)/2}$. The same calculation as for large Macbeath regions implies that the probability of failing to hit some small Macbeath region of $\MM$ is no more than $\eps^{O(n)}$. 

Putting it together, it follows that we can hit the Macbeath regions in all the layers $i$, $0 \leq i \leq k_0$ with failure probability bounded by $2^{-O(n)}$. 

Finally, we describe phase $k_0+1$ of the enumerator. Recall that $\Lambda_{k_0+1}$ consists of points such that the associated minimum volume cap has width at least $\beta$, where $\beta$ is a constant. Let $X$ be a $(K_{\eps},\Lambda_{k_0+1},4c)$-MNet and let $\MM = \MM_{K_{\eps}}^{1/4c}(X)$ be the associated covering. By Lemma~\ref{lem:wide-cap}, the Macbeath regions of $\MM$ have relative volume at least $2^{-O(n)}$. Thus, we can hit all the Macbeath regions of $\MM$ with $2^{O(n)}$ uniformly sampled points in $K$ with failure probability no more that $2^{-O(n)}$. 

In closing, we mention that Lemma~\ref{lem:oracle} shows that the enumerator can construct the three membership oracles it needs for its operation. Specifically, for each point in the hitting set, by part (ii), we can construct an oracle for the associated Macbeath region. By part (i), we can  construct an oracle for $K^*$, which we need to sample uniformly in $K^*$, and by part (iii), we can construct oracles for the caps of $K$ which need to be sampled uniformly. This completes the proof.
\end{proof}

Our algorithm and its analysis follows the general structure presented by Eisenbrand \etal~\cite{EHN11} and Nasz{\'o}di and Venzin~\cite{NaV22}. We solve the $(1+\eps)$-CVP in the norm $\|\cdot\|_K$ by reducing it to the $(1+\eps)$-gap CVP problem in this norm. In the $(1+\eps)$-gap CVP problem, given a target $t$ and a number $\gamma > 0$, we have to either find a lattice vector whose distance to $t$ is at most $\gamma$ or assert that all lattice vectors have distance more than $\gamma / (1+\eps)$. We solve the $(1+\eps)$-CVP problem via binary search on the distance from the target. Given the problem parameters $n$, $\eps$, $\rho = \frac{r'}{r}$, and letting $b$ denote the number of bits in the numerical inputs, the number of different distance values that need to be tested can be shown to be $O(\log n + \log\frac{1}{\eps} + \log \rho + \log b)$. Let $\Phi(n, \eps, \rho, b)$ denote this quantity. For each distance, we need to solve the $(1+\eps)$-gap CVP problem. In turn, the $(1+\eps)$-gap CVP problem is solved by invoking the $(c,\eps)$-enumerator. For each of the $N$ bodies generated by the enumerator, we need to call a 2-gap CVP solver. For this purpose, we use Dadush and Kun's deterministic algorithm~\cite{DaK16} as the 2-gap CVP solver. As this 2-gap CVP solver always yields the correct answer, the only source of error in our algorithm arises from the fact that a valid covering may not be generated. The failure rate of our $(c,\eps)$-enumerator is $2^{-O(n)}$, which we reduce further by running it $\log \Phi(n, \eps, \rho, b)$ times. This ensures that all the coverings generated over the course of solving the $(1+\eps)$-CVP problem are correct with probability at least $1-2^{-O(n)}$. Recalling that the algorithm by Dadush and Kun takes $2^{O(n)}$ time and $O(2^n)$ space, we have established Theorem~\ref{thm:cvp} (neglecting polynomial factors in the input size).

\subsection{Approximate Integer Programming} \label{s:apps-ip}

Through a reduction by Dadush, our CVP result also implies a new algorithm for approximate integer programming (IP). We are given a convex body $K \subseteq \RE^n$ and an $n$-dimensional lattice $L \subset \RE^n$, and we are to determine either that $K \cap L = \emptyset$ or return a point $y \in K \cap L$. The best algorithm known for this problem takes $n^{O(n)}$ time~\cite{Kan87}, which has sparked interest in the approximate version. In approximate integer programming, the algorithm must return a lattice point in $(1+\eps)K$ (where the $(1+\eps)$-expansion of $K$ is about the centroid), or assert that there are no lattice points in $K$.

Dadush~\cite{Dad14} has shown that approximate IP can be reduced to $(1+\eps)$-CVP problem under a well-centered norm. His method is to first find an approximate centroid $p$ and then make one call to a $(1+\eps)$-CVP solver for the norm induced by $K-p$. By plugging in our solver, we obtain an immediate improvement with respect to the $\eps$-dependencies (neglecting polynomial factors in the input size).

\begin{theorem} \label{thm:approx-ip}
There exists a $2^{O(n)}/\eps^{(n-1)/2}$-time and $O(2^{n})$-space randomized algorithm which solves the approximate integer programming problem with probability at least $1 - 2^{-n}$ .
\end{theorem}

\section{Conclusions} \label{s:conc}

In this paper we have demonstrated the existence of concise coverings for convex bodies. In particular, we have shown that given a real parameter $0 < \eps \leq 1$ and constant $c \geq 2$, any well-centered convex body $K$ in $\RE^n$ has a $(c,\eps)$-covering for $K$ consisting of at most $2^{O(n)} / \eps^{(n-1)/2}$ centrally symmetric convex bodies. This bound is optimal with respect to $\eps$-dependencies. Furthermore, we have shown that the size of the covering is instance-optimal up to factors of $2^{O(n)}$. Coverings are useful structures. One consequence of our improved coverings is a new (and arguably simpler) construction of $\eps$-approximating polytopes in the Banach-Mazur metric. We have also demonstrated improved approximation algorithms for the closest-vector problem in general norms and integer programming.

 In contrast to earlier approaches, our covering elements are based on scaled Macbeath regions for the body $K$. This raises the question of what is the best choice of covering elements. Eisenbrand \etal~\cite{EHN11} showed that the size of any covering based on ellipsoids grows as $\Omega(n^{n/2})$, even when the domain being covered is a hypercube. Our Macbeath-based approach results in a reduction of the dimensional dependence to $2^{O(n)}$ for any convex body. Macbeath regions have many nice properties, including the fact that it is easy to construct membership oracles from a membership oracle for the original body. Unfortunately, Macbeath regions have drawbacks, including the fact that their boundary complexity can be as high as $K$'s boundary complexity.

It is natural to wonder whether we can do better than ellipsoid-based coverings with uniform covering elements. For example, can we build more economical coverings based on affine transformations of some other fixed convex body. Recent results from the theory of volume ratios imply that this is not generally possible. The work of Galicer, Merzbacher, and Pinasco \cite{GMP21} (combined with polarity) implies that for any convex body $L$, there exists a convex body $K$, such that for any affine transformation $T$, if $T(L)$ is contained within $K$, then $\vol(T(L))$ is at most $\vol(K) / (b n)^{n/2}$, where $b$ is an absolute constant. A straightforward packing argument implies that if we restrict covering elements to affine images of a fixed convex body, the worst-case size of a $(c,\eps)$ covering grows as $\Omega(n^{n/2})$ (independent of $\eps$).


\pdfbookmark[1]{References}{s:ref}
\setlength{\bibitemsep}{1ex} 
\DeclareFieldFormat[article, inproceedings, incollection]{title}{#1} 
\printbibliography 

\end{document}